\documentclass[aps,prl,singlecolumn,superscriptaddress,nofootinbib,a4paper,longbibliography]{revtex4-2}

\usepackage{times}
\usepackage{amsmath,amssymb,amsthm}

\usepackage{amsfonts,amscd,mathtools,slashed,mathrsfs}
\usepackage{braket}

\usepackage{graphicx}
\usepackage{dcolumn}
\usepackage{bm}
\usepackage{marvosym} 
\usepackage{hyperref}
\usepackage{upgreek}
\usepackage{enumerate}
\usepackage{verbatim}
\usepackage{float}

\makeatletter
\usepackage{tikz}
\newcommand*\circled[2][1.6]{\tikz[baseline=(char.base)]{
    \node[shape=circle, draw, inner sep=1pt, 
        minimum height={\f@size*#1},] (char) {\vphantom{WAH1g}#2};}}
\makeatother

\newtheorem{Thm}{Theorem}
\newtheorem{Cor}[Thm]{Corollary}

\newtheorem{Prop}[Thm]{Proposition}
\newtheorem{Def}[Thm]{Definition}

\newtheorem{fact}[Thm]{Fact}

\newcommand{\itSigma}{\mathit{\Sigma}}

\newcommand{\R}{{\mathbb R}}

\newcommand{\B}{\mathcal{B}}

\newcommand{\K}{\mathcal{K}}

\newcommand{\M}{\mathcal{M}}

\newcommand{\X}{\mathcal{X}}

\newcommand{\A}{\mathcal{A}}
\newcommand{\vc}{\vcentcolon =}             
\newcommand{\cv}{= \vcentcolon}             

\newcommand\utimes{\mathbin{\ooalign{$\cup$\cr%
   \hfil\raise0.42ex\hbox{$\scriptscriptstyle\times$}\hfil\cr}}}
\newcommand\bigutimes{\mathop{\ooalign{$\bigcup$\cr%
   \hfil\raise0.36ex\hbox{$\scriptscriptstyle\boldsymbol{\times}$}\hfil\cr}}}

\newcommand{\bx}{\boldsymbol{x}}

\newcommand{\ba}{\boldsymbol{a}}
\newcommand{\bA}{\boldsymbol{A}}
\newcommand{\bX}{\boldsymbol{X}}

\newcommand{\bp}{\boldsymbol{p}}
\newcommand{\bq}{\boldsymbol{q}}

\DeclareMathOperator{\idle}{idle}
\DeclareMathOperator{\doo}{do}

\newcommand{\Gl}{\mathcal{G}^\text{loop}}

\newcommand{\I}{\mathcal{I}}

\usepackage{color}
\usepackage{xcolor}
\usepackage{framed}

\definecolor{shadecolor}{gray}{0.8}
\definecolor{lgray}{gray}{0.5}

\newcounter{mnotecount}[section]
\renewcommand{\themnotecount}{\thesection.\arabic{mnotecount}}
\newcommand{\mnote}[1]
{\protect{\stepcounter{mnotecount}}$^{\mbox{\footnotesize
$
\bullet$\themnotecount}}$ \marginpar{
\raggedright\tiny\em
$\!\!\!\!\!\!\,\bullet$\themnotecount: #1} }

\definecolor{darkgreen}{rgb}{0,.5,0}

\newcommand{\Pbox}[2]{P \big( #1 \, \big\vert \, #2 \big)}
\newcommand{\PboxR}[3]{P_{#3} \big[ #1 \, \big\vert \, #2 \big]}

\begin{document}


\title{The operational no-signalling constraints and their implications}

\author{Micha\l\ Eckstein}
\email{michal.eckstein@uj.edu.pl}
\affiliation{Institute of Theoretical Physics, Faculty of Physics, Astronomy and Applied Computer Science, 
Jagiellonian University, ul.\ {\L}ojasiewicza 11, 30–348 Krak\'ow, Poland}
\affiliation{International  Centre  for  Theory  of  Quantum  Technologies,  University  of  Gda\'nsk, 
Jana Ba\.zy\'nskiego 8, 80-309
Gda\'nsk,  Poland}

\author{Tomasz Miller}
\affiliation{Copernicus Center for Interdisciplinary Studies, Jagiellonian University,
Szczepa\'nska 1/5, 31-011 Krak\'ow, Poland}
\affiliation{International  Centre  for  Theory  of  Quantum  Technologies,  University  of  Gda\'nsk, 
Jana Ba\.zy\'nskiego 8, 80-309
Gda\'nsk,  Poland}

\author{Ryszard Horodecki}
\affiliation{National Quantum Information Centre in Gda\'nsk, Jana Ba\.zy\'nskiego 8, 80-309 Gda\'nsk, Poland}
\affiliation{International  Centre  for  Theory  of  Quantum  Technologies,  University  of  Gda\'nsk, 
Jana Ba\.zy\'nskiego 8, 80-309
Gda\'nsk,  Poland}
\author{Ravishankar Ramanathan}
\affiliation{School of Computing and Data Science, The University of Hong Kong, Pokfulam Road, Hong Kong}
\author{Pawe\l\ Horodecki}
\affiliation{International  Centre  for  Theory  of  Quantum  Technologies,  University  of  Gda\'nsk, 
Jana Ba\.zy\'nskiego 8, 80-309
Gda\'nsk,  Poland}
\affiliation{Faculty of Applied Physics and Mathematics, National Quantum Information Centre,
Gda\'nsk University of Technology, Gabriela Narutowicza 11/12, 80-233 Gda\'nsk, Poland}

\date{\today}

\begin{abstract}
The study of quantum correlations within relativistic spacetimes, and the consequences of relativistic causality on information processing using such correlations, has gained much attention in recent years. In this paper, we establish a unified framework in the form of operational no-signalling constraints to study both nonlocal and temporal correlations within general relativistic spacetimes. We explore several intriguing consequences arising from our framework. Firstly, we show that the violation of the operational no-signalling constraints in Minkowski spacetime implies either a logical paradox or an operational infringement of Poincar\'{e} symmetry. We thereby examine and subvert recent claims in [\href{https://journals.aps.org/prl/abstract/10.1103/PhysRevLett.129.110401}{\textit{Phys. Rev. Lett.} \textbf{129}, 110401 (2022)}] on the possibility of witnessing operationally detectable causal loops in Minkowski spacetime.  Secondly, we explore the possibility of jamming of nonlocal correlations, controverting a recent claim in [\href{https://www.nature.com/articles/s41467-024-54855-1}{\textit{Nat. Comm.} \textbf{16}, 269 (2025)}] that a physical mechanism for jamming would necessarily lead to superluminal signalling. Finally, we show that in black hole spacetimes certain nonlocal correlations under and across the event horizon can be jammed by any agent without spoiling the operational no-signalling constraints.
\end{abstract}

\maketitle

\section{Introduction}

Einstein's theory of relativity induces fundamental constraints on information processing in space and time. Within classical physics, they were implemented as a demand that all physical influences must flow within the future light-cones. This principle was challenged by the discovery of nonlocal correlations implied by quantum mechanics \cite{EPR}. It sparked the famous Einstein--Bohr debate \cite{EPR,BohrEPR,AspectBell2015} and inspired the influential Bell's theorem \cite{BellThm}, which demonstrates the incompatibility of certain quantum correlations between spacelike separated systems with classical causal explanations \cite{Bell_Nonlocal}. 

Nowadays, the topic of nonlocal correlations in spacetime is still very much at the focus of many studies in both fundamental and applied physics. One line of research investigates classical \cite{WoodSpekkens2015,Cavalcanti18,Pearl21,ColbeckVilasini25} and quantum \cite{QuantumCausality,QuantumCausSpekkens} causal models with the common goal of understanding and exploiting nonlocal quantum correlations. A related field focuses on information-processing tasks, such as device-independent generation of random numbers \cite{DI_randomness}, cryptographic key distribution \cite{QCryptography} or reduction of communication complexity \cite{Quantum_comm_complex}, within quantum (and `super-quantum' \cite{PR_box}) nonlocal theories. On the foundational side, it is of primary importance to understand the properties and implications of quantum correlations in strong gravitational fields near black holes \cite{Harlow,Marolf,Maldacena20}, within the far-sighted programme of unifying quantum mechanics with general relativity.

In spite of extensive studies on nonlocal correlations, some key questions remain open: What are the necessary and sufficient conditions for a physical theory to exclude the possibility of superluminal information transfer \cite{PR_box,PawelRaviCausality,Loops_PRL}? Does relativistic physics allow for cyclic influences \cite{Novikov2,DeutschCTC,Loops_PRL}? Is it possible to change nonlocal correlations at a distance without superluminal signalling \cite{PawelRaviCausality,Monogamy_HR,ColbeckVilasini25}? Are there device-independent cryptographic protocols, which remain secure in any physical theory \cite{NS_crypto,RC_crypto}?

Here we lay the groundwork to answer the above questions and more by putting forward the \emph{operational no-signalling constraints} and exploring their consequences. To this end, we establish a unified rigorous framework to handle spacelike and timelike correlations in general spacetimes. It is purely operational and based solely on spacetime geometry and the concept of a spacetime random variable \cite{SRV}. 
Our framework generalises the notion of an input to a general random variable, encompassing both the notion of measurement setting in a Bell-type experiment \cite{Bell_Nonlocal} and that of an intervention in causal modelling \cite{Pearl,ColbeckVilasini25}. In such a general context, we identify the \textit{necessary and sufficient} conditions for the impossibility of operational superluminal signalling under the natural assumption of agents' freedom of selecting their inputs. 


The operational no-signalling constraints grounded in the established formalism prove to be a powerful tool to resolve some of the fundamental questions concerning nonlocal correlations in spacetime. 

Firstly, we show that its violation in Minkowski spacetime would either lead to a logical paradox or to an operational violation of the Poincar\'e symmetry. This undermines the validity of some recent claims on the possibility of witnessing operationally detectable causal loops in Minkowski spacetime \cite{Loops_PRA,Loops_PRL,ColbeckVilasini25}. On the other hand, in a curved spacetime, this implication holds only locally. This insight introduces a new twist in the study of nonlocal correlations, as it shows that the role of spacetime cannot be reduced to the relativistic causal structure. 

Secondly, we explore the intriguing possibility of \emph{jamming} of nonlocal correlations, where one party, Bob, is able to influence the correlations between measurement outcomes of two space-like separated parties, Alice and Charlie, while at the same time not influencing their individual measurement statistics \cite{Jamming,PawelRaviCausality}. Recently, it has been claimed --- basing on certain monogamy relations for nonlocal correlations --- that a physical mechanism for jamming would necessarily lead to superluminal signalling \cite{Monogamy_HR}. We refute this claim by showing that diverse jamming scenarios are compatible with the operational no-signalling constraints. In particular, we show that there exist spacetime configurations, in which an agent can influence nonlocally the $n$-party correlations for an arbitrarily large $n$, without affecting any $k$-party correlation for any $k < n$. This implies that jamming could, in principle, be used to attack general device-independent cryptographic protocols based on nonlocal correlations in these spacetime configurations. 
We also derive general monogamy relations in signalling physical theories and clarify that the existence of such monogamies should be interpreted as frustrations caused by opposing constraints, rather than limitations imposed by the underlying physics.

Finally, we show that in black hole spacetimes the operational no-signalling constraints allow for jamming of certain nonlocal correlations below or across the event horizon by any agent. We conclude with a discussion of the implications of our findings for quantum information processing in general spacetimes, and questions for future work.

\section{Results}

\subsection{Basic concepts and notation}


We consider a general `black-box' experimental scenario with $n$ agents acting on a shared physical system. Every agent has an input and an output, which are denoted by two random variables (RV): $X_i$ taking values in some finite set $\X_i$ and $A_i$ with values in $\A_i$, respectively. The agents experiment on a physical system at hand by studying various statistics of the outputs, conditioned upon the chosen inputs. In general, the RVs can be of different types, depending on the details of the agents' physical interaction with the system and the experimental context. In particular, we can consider `intervening' agents which may have no outputs, $\A=\emptyset$ 
and `passive' agents, which need not have inputs, $\X = \{\text{idle}\}$.

%


%
We further assume that the scenario is embedded in a relativistic spacetime $\M$, which comes equipped with a binary relation $\prec$, called the (strict) causal relation. For any point $p \in \M$ it determines the subsets of $\M$ called the \emph{future} of $p$ and the \emph{past} of $p$, defined respectively as $R^+(p) = \{ q \in \M \, |\, p \prec q\}$ and $R^-(p) = \{ q \in \M \, |\, q \prec p\}$. It is customary to denote $p \preceq q$ if $p \prec q$ or $p =q$. See Methods for further precise definitions.

Consequently, we promote any RV to a \emph{spacetime random variable} (SRV), following the seminal work \cite{SRV}. An SRV is a pair $(X,p)$, where $X$ is an RV and $p$ is a point in the spacetime $\M$. Operationally, the notation $(x,p) \vc (X=x,p)$ designates an event at which the random variable $X$ assumes the value $x$, which happens at the spacetime point $p$. The probability of such an event is denoted as $P_X(x,p) \vc P_{X,p}(X=x,p)$.
%

Since we are dealing with two $n$-tuples of SRVs, it is convenient to use a `vector' notation, where $(\bX,\bp)$ stands for the $n$-tuple $(X_i,p_i)_{i=1}^n$. We shall also consider `subtuples', such as $(X_i,p_i)_{i \in F}$, for some subset of indices $F \subset \{1\ldots,n\}$ and denote them by $(\bX,\bp)^F$. This notation will also be applied to the tuples of spacetime points themselves, $\bp = (p_i)_{i = 1}^n$ and $\bp^F = (p_i)_{i \in F}$. Eventually, an experimental scenario is operationally described by the collection of conditional probabilities
\begin{align}\label{box}
\PboxR{(\ba,\bq)}{(\bx,\bp)}{\bA|\bX}.
\end{align}
The marginals are defined in the standard way:
\begin{align}\label{marginal}
\PboxR{(\ba,\bq)^G}{(\bx,\bp)}{\bA^G|\bX} = 
\sum\limits_{a_k \in \A_k, k \notin G} \PboxR{(\ba,\bq)}{(\bx,\bp)}{\bA|\bX},
\end{align}
for any $G \subset \{1\ldots,n\}$.

The standard Bell scenario for spacelike separated agents with local measurement settings and outcomes is a special case of \eqref{box} with $q_i = p_i$ and $p_i \nprec p_j$, for all $i,j$ (see SI for more details). Our framework extends it in three aspects, which are natural from the physical perspective.

Firstly, while the agents' outputs are always associated with the outcomes of some measurements, the inputs are not limited to be the measurement settings. The `input' can describe operationally any influence of an agent on the physical system at hand (cf. e.g. \cite{Gisin2020NS} or \cite{NJP2025}). This includes a measurement with settings, but also, for instance, the use of a tunable quantum gate applied locally to a quantum system or a directed magnetic field. 
More abstractly, an agent can make an `intervention' on the system, forcing a given RV to take a desired value (cf. \cite{ColbeckVilasini25,Loops_PRL}) --- see Methods for the formal discussion.%

Secondly, we recognize the fact that the input and output SRVs of an agent never happen exactly at the same spacetime point, although, clearly, the output must be in the input's future, $p_i \prec q_i$. For instance, in some loophole-free Bell tests \cite{Bell_2015_1,Bell_delay}, the effective delay between the choice of the settings and the generation of the output was typically around 3 \textmu s. It is precisely the existence of this delay that coerced the need to separate Alice's and Bob's detectors by over 1 km in order to close the locality loophole.

Thirdly, we do not assume that the agents are mutually spacelike separated --- some of them can be in the others' future or past. We thus treat both spacelike and timelike correlations on the same footing and within one general framework.

In summary, apart from the obvious demand $p_i \prec q_i$ for all $i$ we do not impose any \emph{a priori} relations between the spacetime location of agents' actions. Consequently, the collection of probabilities \eqref{box} describes the most general black-box experimental scenario in a relativistic spacetime. It includes, in particular, all scenarios with interventions studied recently in \cite{Ringbauer16,Loops_PRA,Loops_PRL,Grothus24,ColbeckVilasini25}.

\subsection{Operational no-signalling constraints}

The spacetime structure induces restrictions on the admissible correlations of the form \eqref{box}. To understand these constraints, the following concepts will prove useful (see Fig. \ref{fig:OS}):

\begin{Def}\label{opsepDef}
    For a tuple of spacetime points $\bq$ we define their \emph{gathering point} as any point $Q$, such that $q_j \preceq Q$ for all $j$. 
    We say that a tuple $\bq$ of spacetime points is \emph{operationally separated} from a tuple $\bp$ if there exists a gathering point $Q$ for $\bq$, which lies outside of the future of all $p_i$'s, that is $p_i \nprec Q$ for all $i$.
    
\end{Def}

The existence of the maximal speed of signal propagation --- the speed of light in vacuum --- implies that information from spacelike separated laboratories must first be gathered before it is compared and jointly processed. More formally, if we have two SRVs $(A_1,q_1)$ and $(A_2,q_2)$, then a common SRV $((A_1,A_2),Q)$ is defined only if $Q$ is a gathering point for $q_1$ and $q_2$. This stems from the fact that there is no SRV if there is no agent who could read it out.

The operational separation relation differs from the standard \emph{spacelike separation relation} involved in Bell-type scenarios in two key aspects: Firstly, it is asymmetric --- we can have $\bq$ operationally separated from $\bp$, but $\bp$ not operationally separated from $\bq$. Secondly, if $\bq$ is operationally separated from $\bp$ then, for any pair of indices $i,j$, the point $q_j$ is either spacelike separated from $p_i$ or it is in $p_i$'s past. These features capture the time-directionality inherent in information processing protocols in spacetime. See Methods and SI for the technical details. 

\begin{figure}[H]
\begin{center}
\resizebox{0.6\textwidth}{!}{\includegraphics[scale=1]{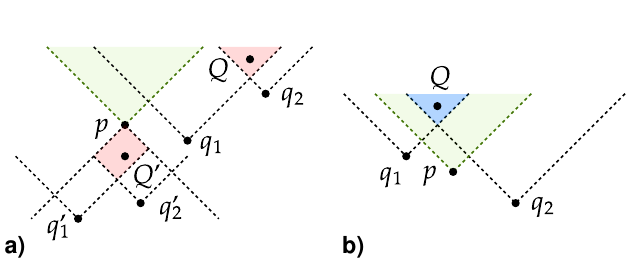}} 
\caption{\label{fig:OS}An illustration of the operational separation relation. In panel \textbf{a)} all points $q_1,q_2,q_1',q_2'$ are operationally separated from the point $p$. Also, both tuples $(q_1,q_2)$ and $(q_1',q_2')$ are operationally separated from $p$, because there exist gathering points, $Q$ and $Q'$ respectively, which lie outside of $p$'s future. On the other hand, $p$ is operationally separated from $(q_1,q_2)$, but it is not operationally separated from $(q_1',q_2')$. In panel \textbf{b)} the tuple $(q_1,q_2)$ is not operationally separated from $p$, although both single points $q_1$ and $q_2$ are operationally separated from $p$. Also, $p$ is operationally separated from $q_1$, and $q_2$, and $(q_1,q_2)$.}
\end{center}
\end{figure}


We are now in a position to formulate the announced no-signalling constraints.

\begin{Def}
\label{opcausDef}
We say that correlations of the form \eqref{box} satisfy the \emph{no-signalling constraints} when for any $F, G \subset \{1,\ldots,n\}$:
\begin{equation}\label{ONS_new}
\begin{aligned}
& \text{ if }\bq^G \text{ is operationally separated from } \bp^F, \text{ then}
\\
& \PboxR{(\ba,\bq)^G}{(\bx,\bp)}{\bA^G|\bX} = \PboxR{(\ba,\bq)^G}{(\bx',\bp)}{\bA^G|\bX},
\end{aligned}
\end{equation}
for any $\ba$, any $\bx$ and any $\bx'$ such that $x'_i = x_i$ for all $i \notin F$.
\end{Def}

In words, if the subset of output points $\bq^G$ is operationally separated from the subset of input points $\bp^F$, then manipulating the inputs pertaining to $ F$ cannot influence the output statistics pertaining to $G$. As we will now show, these constraints rule out the possibility of operational superluminal signalling --- hence the name.

\medskip

The operational notion of ``signalling'' is based on the classical concept of information transfer \cite{Shannon}: The input RV can be employed by an agent to encode a ``\emph{message}'', which is then processed by a (possibly nonlocal) physical system and decoded from the statistics of the output RVs. The message itself is freely selected, from a set of possible messages, by the sending agent. The \emph{free selection} is implemented in our framework by demanding that an agent can force a chosen value of an RV, $X = x$, independently of any external influences. In the language of causal modelling \cite{Pearl}, such an operation corresponds to an \emph{intervention} on an RV (see Methods).

\begin{Thm}\label{thm:MAIN}
Under the assumption that the input RVs can be freely selected, the no-signalling constraints \eqref{ONS_new} are necessary and sufficient for the lack of operational superluminal signalling.
\end{Thm}

If \eqref{ONS_new} is violated, then the proof consists in showing that there exists an SRV $(B,Q)$, with $B = f(\ba^G)$ for some non-constant function $f$ and an index $j \in F$, such that 
\begin{align}
\PboxR{(b,Q)}{(\bx,\bp)}{B|\bX} \neq \PboxR{(b,Q)}{(\bx',\bp)}{B|\bX},
\end{align}
where $x'_j \neq x_j$ and $x'_i = x_i$ for $i\neq j$. Because $Q$ is a gathering point for $\bq^G$ and $\bq^G$ is operationally separated from $\bp^F$, then $p_j \nprec Q$. Consequently, a sending agent freely selecting the input $X_j = x_j$ or $X_j = x_j'$ at the point $p_j$ can change the local detection statistic of the RV $B$ registered by an agent at $Q$ outside of the sender's future.

The complete general proof of Thm. \ref{thm:MAIN} is presented in SI. Below, we illustrate the signalling protocol for the case $n=4$ and under the assumption that \eqref{ONS_new} is violated for $F=\{1,4\}$ and $G=\{2,3\}$. Concretely, let us assume that
\begin{equation}
\begin{aligned}
     &\PboxR{(a_2,q_2),(a_3,q_3)}{(x_1,p_1),(x_2,p_2),(x_3,p_3),(x_4,p_4)}{A_2,A_3|\bX}\\
    \neq \; & \PboxR{(a_2,q_2),(a_3,q_3)}{(x_1',p_1),(x_2,p_2),(x_3,p_3),(x_4,p_4)}{A_2,A_3|\bX},
\end{aligned}
\end{equation}
with the spacetime configuration as in Fig. \ref{fig:thm}. An agent following a timelike curve (thick blue line) firstly prepares the inputs $x_2$, $x_3$ and $x_4$ at the orchestrating point $p_0$ and then encodes his bit by choosing the input $x_1$ or $x_1'$ at point $p_1$. By doing so, he influences the joint output statistics of RVs $A_2$ and $A_3$, which can be read out by an agent at the gathering point $Q$. In this way, statistical superluminal information transfer is performed from point $p_1$ in spacetime to point $Q$ outside of $p_1$'s future. Note that the outcomes at $q_1$ and $q_4$ do not play any role, but the selection of concrete inputs $x_2$, $x_3$, $x_4$ is necessary.

\begin{figure}[H]
\begin{center}
\resizebox{0.6\textwidth}{!}{\includegraphics[scale=1]{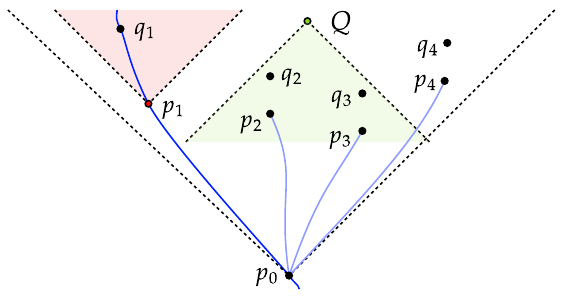}} \qquad
\caption{\label{fig:thm} An illustration of the signalling protocol exploiting a violation of constraint \eqref{ONS_new}. See text for the description.}
\end{center}
\end{figure}

\subsection{Jamming of nonlocal correlations}

The fact that delocalised information in spacetime needs to be gathered leads to an intriguing effect known as \emph{jamming} of nonlocal correlations \cite{Jamming,PawelRaviCausality}. The simplest scenario involves three agents in a spacetime configuration as in Fig. \ref{fig:OS} \textbf{b)}. Alice and Charlie read out the SRVs $(A_1,q_1)$ and $(A_2,q_2)$, respectively, while Bob is free to select the value of the SRV $(X,p)$. The no-signalling constraints \eqref{ONS_new} forbid Bob's influence on Alice's and Charlie's local statistics,
\begin{align}\label{jam0}
    \PboxR{(a_1,q_1)}{(x,p)}{A_1 | X} = \PboxR{(a_1,q_1)}{(x',p)}{A_1 | X} \quad \text{ and } \quad \PboxR{(a_2,q_2)}{(x,p)}{A_2 | X} = \PboxR{(a_2,q_2)}{(x',p)}{A_2 | X},
\end{align}
for all values $a_1,a_2, x, x'$. On the other hand, because the tuple $(q_1,q_2)$ is \emph{not} operationally separated from $p$, Bob's influence on correlations between Alice's and Charlie's outcomes is allowed,
\begin{align}\label{jam1}
    \PboxR{(a_1,q_1),(a_2,q_2)}{(x,p)}{A_1 A_2 | X} \neq \PboxR{(a_1,q_1),(a_2,q_2)}{(x',p)}{A_1 A_2 | X},
\end{align}
for some values $a_1,a_2, x \neq x'$. Correlations with properties \eqref{jam0} and \eqref{jam1} are known in the literature as `\emph{relativistically causal}'. Theorem \ref{thm:MAIN} shows unequivocally that such an effect does not facilitate operational superluminal signalling,  in contradiction to some recent claims \cite{Monogamy_HR} --- these are analysed in detail below.

The possibility of jamming has been shown to have devastating consequences for the security of certain device-independent (DI) cryptographic protocols \cite{PawelRaviCausality,RC_crypto}. Here we strengthen these results by showing that, in some specific spacetime configurations, a relativistic adversary can successfully attack a DI cryptographic protocol based on nonlocal correlations between an arbitrarily large number of parties.

\begin{Thm}\label{thm:jam}
    Let $\M$ be a spacetime with at least two spatial dimensions. For any $n \geq 2$, it is always possible to arrange a configuration of spacetime points $p, q_1, \ldots, q_n$ in such a way that $\bq$ is not operationally separated from $p$, but $\bq^G$ is operationally separated from $p$ for any subset $G \subsetneq \{1,\ldots,n\}$.
\end{Thm}

This result is proven in SI and illustrated in Fig. \ref{fig:njam}. Suppose now that we have a collection of conditional probabilities $\PboxR{(\ba,\bq)}{(x,p)}{\bA | X}$ for $n+1$ parties\footnote{For sake of simplicity we assume that the reading agents do not have any inputs and the jammer has no output. These features can be added to the protocol without affecting the main conclusion.}, in a spacetime configuration as in Thm. \ref{thm:jam}, with the properties
\begin{align}
    \PboxR{(\ba,\bq)^G}{(x,p)}{\bA^G | X} = \PboxR{(\ba,\bq)^G}{(x',p)}{\bA^G | X},
\end{align}
for all $G \subsetneq \{1,\ldots,n\}$, $\ba^G$, $x$ and $x'$,
but 
\begin{align}
    \PboxR{(\ba,\bq)}{(x,p)}{\bA | X} \neq \PboxR{(\ba,\bq)}{(x',p)}{\bA | X},
\end{align}
for some $\ba$ and $x \neq x'$. Thm. \ref{thm:jam}, together with Thm. \ref{thm:MAIN}, show that such a scenario does not facilitate any operational superluminal signalling. An agent with an operational access to the SRV $(X,p)$ can freely affect the $n$-party correlations, without changing any $k$-party correlations for $k < n$. It means that all $n$ agents must meet and compare their outcomes to check whether someone has tampered with their correlations. This leads to the announced conclusion:
\begin{Cor}\label{cor:jam}
   For spacetime configurations as in Thm. \ref{thm:jam} any device-independent cryptographic protocol for random number generation or key distribution based on nonlocal correlations between an arbitrarily large number of parties can be successfully attacked (in the sense of \cite{PawelRaviCausality,RC_crypto}) by a relativistic adversary.
\end{Cor}

\medskip

In Minkowski spacetime a gathering point exists for any collection of spacetime points (see Fig. \ref{fig:BH} \textbf{a)}), but, interestingly, it should be noted that in a general spacetime this need not be the case. An example is provided by the Schwarzschild black hole spacetime. Because of the spacetime singularity, there exist tuples of points, which do not have a gathering point --- see Fig. \ref{fig:BH} \textbf{b)}. This means that they are \emph{not} operationally separated from any other point in spacetime. It leads to the following result:

\begin{Prop}\label{prop:BH}
   In black hole spacetimes there exist spacelike separated SRVs, such that any agent can change correlations between them, without effectuating superluminal signalling. 
\end{Prop}

\begin{figure}[H]
\begin{center}
\resizebox{0.6\textwidth}{!}{\includegraphics[scale=1]{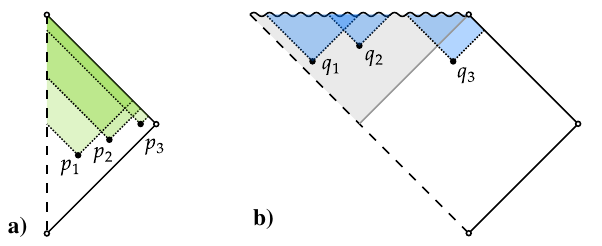}}
\caption{\label{fig:BH} Conformal diagrams for the Minkowski, \textbf{a)}, and Schwarzschild, \textbf{b)}, spacetimes (see e.g. \cite{ChruscielBH}). \textbf{a)} In Minkowski spacetime any tuple of points have a non-empty common future, and hence a gathering point. \textbf{b)} The tuple $(q_1,q_2)$ has a gathering point, but this is not the case for the tuples $(q_1,q_3)$, $(q_2,q_3)$ and $(q_1,q_2,q_3)$.}
\end{center}
\end{figure}

\subsection{Special principle of relativity and causal loops}

It is commonly believed that operational superluminal signalling leads to causal loops and the ensuing logical paradoxes (see e.g. \cite{PawelRaviCausality,PRA2020}). We now clarify this intuition by leveraging the \emph{special principle of relativity}, which says that the laws of physics ought to be the same in all inertial reference frames \cite{Einstein_relativity}. This demand is implemented in our formalism as follows: Assume that $\M$ is the Minkowski spacetime and let $L: \M \to \M$ denote an (active) Poincar\'e transformation. Using the vector notation we write $L(\bp) = \{L(p_i)\}_{i=1}^n$.

\begin{Def}
    We say that a collection of probabilities \eqref{box} in Minkowski spacetime \emph{satisfies the special principle of relativity} when
\begin{align}
& 
\PboxR{(\ba,\bq)}{(\bx,\bp)}{\bA|\bX} = \PboxR{(\ba,\bq)}{(\bx',\bp)}{\bA|\bX} \notag \\
& \text{ if and only if} \label{Poincare}
\\
& \PboxR{(\ba,L(\bq))}{(\bx,L(\bp))}{\bA|\bX} = \PboxR{(\ba,L(\bq))}{(\bx',L(\bp))}{\bA|\bX}, \notag
\end{align}
for any Poincar\'e transformation $L$.
\end{Def}

Suppose that \eqref{Poincare} is violated for some transformation $L$. Then, there exist two inertial frames linked by $L$ that can be \emph{operationally} distinguished in a suitable experiment. 
In other words, if information transfer (sub- or superluminal) is possible in one frame, then it is possible in any other frame. Notwithstanding, the probabilities of signal detection are, in general, frame-dependent.

In this context, Theorem \eqref{thm:MAIN} has an immediate consequence:
\begin{Cor}\label{cor:loops}
    Let $\M$ be the Minkowski spacetime. Then, the violation of condition \eqref{ONS_new} leads either to logical paradoxes or to the violation of the special principle of relativity.
\end{Cor}

If condition \eqref{ONS_new} is violated in one inertial frame, then --- by the special principle of relativity --- it must be violated in any other inertial frame. Consequently, one can construct an operational causal loop, which can be exploited by an agent to arrive at a self-contradiction --- see Fig. \ref{fig:loop}. Indeed, suppose that constraint \eqref{ONS_new} is violated in such a way that an agent can influence the joint statistics of SRVs $(A_1,q_1)$, $(A_2,q_2)$ by intervening on an SRV $(X,p)$, as in Eq. \eqref{jam1}. 
There exists a Poincar\'e transformation $L$ with $p' = L(p)$, $q'_i = L(q_i)$, such that $q_1, q_2 \prec p'$ and $q_1', q_2' \prec p$. If the special principle of relativity \eqref{Poincare} holds, then
\begin{align}\label{jam1p}
    \PboxR{(a_1,q_1'),(a_2,q_2')}{(x,p')}{A_1 A_2 | X} \neq \PboxR{(a_1,q_1'),(a_2,q_2')}{(x',p')}{A_1 A_2 | X}.
\end{align}
Now, the agent can design the setup in such a way that his free selection of the value $x$ of the SRV $(X,p)$ yields a value $(X=x',Q')$ with non-zero probability. But since $Q'$ is in his past, this contradicts his free selection.

\begin{figure}[H]
\begin{center}
\resizebox{0.4\textwidth}{!}{\includegraphics[scale=1]{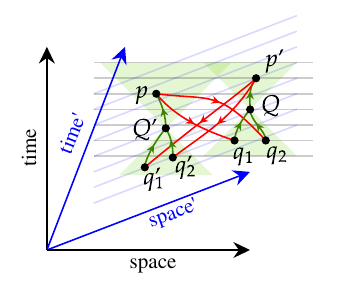}} 
\caption{\label{fig:loop} A spacetime diagram illustrating an operational causal loop exploiting a violation of no-signalling constraints \eqref{ONS_new} and the special principle of relativity \eqref{Poincare}. The arrows represent causal influences between SRVs, subluminal (green) and superluminal (red). See text for further description.}
\end{center}
\end{figure}

This result shows that under natural symmetry assumptions the no-signalling principle does rule out the existence of causal loops in Minkowski spacetime, in contrast to some recent claims \cite{Loops_PRA,Loops_PRL,ColbeckVilasini25} --- see below for a detailed study. However, two comments are in order:


Firstly, recall that any relativistic spacetime can locally be approximated to the Minkowski spacetime. Consequently, Cor.~\ref{cor:loops} applies \emph{locally} in a `sufficiently small' neighbourghood of any free-falling laboratory, modelled with the help of Fermi normal coordinates \cite{Fermi_coordinates}, in any spacetime $\M$. However, the Poincar\'e group is not a group of symmetries of a general relativistic spacetime, so Cor.~\ref{cor:loops} is not expected to hold \emph{globally} in a curved spacetime.

Secondly, Cor.~\ref{cor:loops} forbids the existence of \emph{operational} causal loops, which lead to logical inconsistencies. But it does not say anything about the mechanisms underlying the operational probabilities \eqref{box}. One can, for instance, imagine a theory with cyclic causal influences at the Planck scale, which cannot be operationally exploited by any macroscopic agent. In a similar vein, the maximal extension of the Kerr spacetime contains closed timelike curves under the inner horizon \cite{Carter_CTC}, but these are inaccessible to any agent in the domain of outer communications. Such considerations fall under the slogan of the \emph{self-consistency principle}, which has a long history in the context of closed timelike curves in general relativity \cite{HawkingChronology,Novikov2,ThorneCTC} and quantum information processing \cite{DeutschCTC}.

\subsection{Case study 1: Jamming with the one-time pad}

In \cite{Loops_PRL} the authors  put forward an interesting case. Namely, it was claimed that if one only forbids superluminal signalling while allowing for the possibility of superluminal causal influences, then there arises the mathematical possibility of closed causal loops embeddable in $(1+1)$-Minkowski spacetime that nevertheless do not lead to logical paradoxes. This was done by means of a general framework for causal models developed by the authors in \cite{Loops_PRA}, and by a concrete construction of an ``operationally detectable causal loop'' embeddable in a $(1+1)$-Minkowski spacetime based on a fine-tuned classical causal model.

We reexamine this claim in light of the operational no-signalling conditions of Def. \ref{opcausDef}, which were shown to be necessary and sufficient to rule out operational signalling in Thm. \ref{thm:MAIN}. This is done in detail in SI, while here we show that the causal model $\Gl$ presented in \cite{Loops_PRL} can be safely embedded in Minkowski spacetime (of any dimension), and without the fine-tuning contraints on agents' locations. The consistent embedding is based on the fact that agents' input (`intervention') SRVs are in the \emph{strict} past of the corresponding output SRVs. Concretely, suppose that we have the six SRVs located as in Fig. \ref{fig:CV}.

\begin{figure}[H]
\begin{center}
\resizebox{0.5\textwidth}{!}{\includegraphics[scale=1]{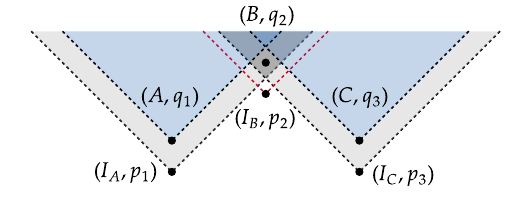}} 
\caption{\label{fig:CV}A consistent embedding of the causal model $\Gl$ from \cite{Loops_PRL} in Minkowski spacetime. See text for the description.}
\end{center}
\end{figure}

The operational no-signalling constraints \eqref{ONS_new} allow for \emph{both}
\begin{align}
    & \PboxR{(b,q_2)}{(\doo(a),p_1),(\doo(c),p_3)}{B|I_A,I_C} \neq \PboxR{(b,q_2)}{(\idle,p_1),(\idle,p_3)}{B|I_A,I_C}, \\
  \text{and } \quad  & \PboxR{(a,q_1),(c,q_3)}{(\doo(b),p_2)}{AC|I_B} \neq \PboxR{(a,q_1),(c,q_3)}{(\idle,p_2)}{AC|I_B},
\end{align}
for some $a,b,c$ and (possibly other) $a',b',c'$. The first inequality means that Alice and Charlie can (jointly) influence Bob's detection statistics --- this is possible, because Bob's readout point, $q_2$, is in the future of both Alice's and Charlie's intervention points, $p_1$ and $p_2$. The second inequality implements the possibility of Bob's jamming of correlations between Alice's and Charlies's SRVs, $(A,q_1)$ and $(C,q_3)$, respectively. This is admissible, because the tuple $(q_1,q_3)$ is \emph{not} operationally separated from $p_2$.

However, in the degenerate case, $p_i = q_i$, considered in \cite{Loops_PRL} the tuple $(q_1,q_3)$ \emph{is} operationally separated from $p_2$. Consequently, the operational no-signalling constraints \eqref{ONS_new} are violated and a logical paradox can occur, unless one denies all the agents' freedom to intervene or not. See SI for the details.

\subsection{Case study 2: Monogamy relations and jamming mechanisms}

In \cite{Monogamy_HR}, an interesting study was undertaken on the monogamy relations of the correlations in relativistic causal theories obeying the constraints from \cite{PawelRaviCausality} (a similar set of constraints to those in Thm. \ref{thm:MAIN}). Specifically, a particular ``triangle'' spacetime configuration was considered \cite[Fig. 2]{Monogamy_HR} with six spacelike separated SRVs $A, B, C, X, Y, Z$ where $X, Y, Z$ are to be treated as inputs and $A, B, C$ are to be treated as outputs and the relativistic causality constraints were given as
\begin{eqnarray}
\label{eq:relcaus-sixconf}
    P_{A,B|X,Y,Z}(a,b|x,y,z) &=& P_{A,B|Z}(a,b|z), \nonumber \\
    P_{A,C|X,Y,Z}(a,c|x,y,z) &=& P_{A,C|Y}(a,c|y), \nonumber \\
    P_{B,C|X,Y,Z}(b,c|x,y,z) &=& P_{B,C|X}(b,c|x), \nonumber \\
    P_{A|X,Y,Z}(a|x,y,z) &=& P_{A}(a), \nonumber \\
    P_{B|X,Y,Z}(b|x,y,z) &=& P_{B}(b), \nonumber \\
    P_{C|X,Y,Z}(c|x,y,z) &=& P_{C}(c).
\end{eqnarray}
That is, in this case $Z$ is in principle able to jam the correlations between $A$ and $B$, similarly $Y$ is in principle able to jam the correlations between $A$ and $C$, and similarly $X$ is in principle able to jam the correlations between $B$ and $C$, while no pairwise affects relations (i.e. point-to-point signalling) exists for any pair of variables.

In the case when all variables are binary, the following interesting entropic monogamy relation was derived \cite{Monogamy_HR} for all correlations obeying \eqref{eq:relcaus-sixconf}:
\begin{eqnarray}
\label{eq:entropic-mono}
    I(AB:Z) + I(AC:Y) + I(BC:X) \leq 1.
\end{eqnarray}
Here $I$ denotes the mutual information defined as $I(M:N) = H(M) + H(N) - H(MN)$ with $H(M) = - \sum_{m} P(m) \log P(m)$ being the Shannon entropy of $M$. The monogamy relation \eqref{eq:entropic-mono} says that when $Z$ has maximal information about $A, B$ such as for instance when $Z$ is used to jam the correlations between $A$ and $B$ as $Z = A \oplus B$, then $I(AC:Y) = I(BC:X) = 0$, i.e., $Y$ is independent of $A,C$ and $X$ is independent of $B,C$. 

Now, an interesting and surprising conclusion was drawn from the existence of monogamy relations such as \eqref{eq:entropic-mono}. Namely, it was noted that the monogamy relations mean that the relativistically causal correlations and thus also potential mechanisms leading to them have to be highly nonlocal. Furthermore, while the correlations themselves do not by definition feature superluminal transfer of information, superluminal signalling may occur if one assumes that there is a jamming mechanism that an agent can turn on and off that leads to these relativistically causal effects. Specifically, turning on or off a jamming device (say $Y$) in this scenario does not directly decide whether jamming occurs (e.g. $Y = B \oplus C$ is set). Instead, the author of \cite{Monogamy_HR} postulates that a highly nonlocal mechanism must exist that takes into account all potential parties that may attempt to jam. In particular, seen from this viewpoint, the author claims that this puts into question the explanation of these correlations via a jamming mechanism and is also problematic for experimentalists probing a given experimental setup of interest since some unknown space-like separated variables may control how they operate their experiment. In conclusion, the \cite{Monogamy_HR} claims to refute suggestions for physical mechanisms that could lead to such relativistically causal correlations and thereby questions the possibility and physicality of jamming. 

In what follows, we controvert this claim and proceed to clarify why the relativistically causal correlations of \cite{PawelRaviCausality} and Thm. \ref{thm:MAIN} may nevertheless be physical. In particular, we point out that such ``monogamy'' or tradeoff relations exist in both signalling and no-signalling theories.

We now illustrate the incompatibility of certain constraints on correlations by re-analysing the ``compass setup'' scenario presented in \cite[Fig. 3 a)]{Monogamy_HR}, which is a sub-scenario. We have four input SRVs $(X,p_1)$, $(X_m,p_1)$, $(Y,p_2)$, $(Y_m,p_2)$ and three output SRVs $(A,q_1)$, $(B,q_1)$, $(C,q_3)$. All points, $p_1, p_2, q_1, q_2, q_3$ are spacelike separated, and hence mutually operationally separated. Furthermore, the tuple $(q_1,q_2)$ is operationally separated from $p_2$, but not from $p_1$, while $(q_2,q_3)$ is operationally separated from $p_1$, but not from $p_2$. Consequently, the operational no-signalling constraints \eqref{ONS_new} read (for sake of readability we omit the spacetime points):
\begin{equation}\label{compass}
\begin{aligned}
 \PboxR{a,b}{x,x_m,y,y_m}{A,B|X,X_m,Y,Y_m} & = \PboxR{a,b}{x,x_m,y',y_m'}{A,B|X,X_m,Y,Y_m},  \\
 \PboxR{b,c}{x,x_m,y,y_m}{B,C|X,X_m,Y,Y_m} & = \PboxR{b,c}{x',x_m',y,y_m}{B,C|X,X_m,Y,Y_m}, \\
 \PboxR{a}{x,x_m,y,y_m}{A|X,X_m,Y,Y_m} & = \PboxR{a}{x',x_m',y',y_m'}{A|X,X_m,Y,Y_m},  \\
 \PboxR{b}{x,x_m,y,y_m}{B|X,X_m,Y,Y_m} & = \PboxR{b}{x',x_m',y',y_m'}{B|X,X_m,Y,Y_m},  \\
 \PboxR{c}{x,x_m,y,y_m}{C|X,X_m,Y,Y_m} & = \PboxR{c}{x',x_m',y',y_m'}{C|X,X_m,Y,Y_m},
\end{aligned}
\end{equation}
for all $a, b, c, x, x', x_m, x_m', y, y', y_m, y_m'$. Now, the author of \cite{Monogamy_HR} assumes that the jamming mechanism, which can be turned on ($X_m = 1$) by Xavier, correlates the outputs of Alice and Bob through Xavier's input, $A \oplus B = X$. If the jamming mechanism is off  ($X_m=0$), then $A$ and $B$ are independent. The same jamming mechanism, controlled by Yanina's RV, $Y_m$, governs the correlations between $B$ and $C$ through Yanina's input $Y$. These assumptions are implemented as
\begin{align}\label{jam_X}
   \PboxR{a,b,c}{x,1,y,y_m}{A,B,C|X,X_m,Y,Y_m} & = \delta_{a \oplus b,x} \big( \lambda \PboxR{c}{0,1,y,y_m}{C|X,X_m,Y,Y_m} \delta_{x,0} + (1- \lambda) \PboxR{c}{1,1,y,y_m}{C|X,X_m,Y,Y_m} \delta_{x,1} \big) , \\
   \PboxR{a,b,c}{x,x_m,y,1}{A,B,C|X,X_m,Y,Y_m} & = \delta_{b \oplus c,y} \big( \mu \PboxR{a}{x,x_m,0,1}{A|X,X_m,Y,Y_m} \delta_{y,0} + (1- \mu) \PboxR{a}{x,x_m,1,1}{A|X,X_m,Y,Y_m} \delta_{y,1} \big),\label{jam_Y} 
\end{align}
for all $a, b, c, x, x_m, y, y_m$. 
But these demands are contradictory, under the assumption of operational no-signalling constraints \eqref{compass}. Indeed, assumption \eqref{jam_X} implies that
\begin{align}\label{jam_X2}
   \PboxR{a,b,c}{0,1,y,1}{A,B,C|X,X_m,Y,Y_m}  & = \delta_{a \oplus b,0} \cdot \lambda \cdot \PboxR{c}{0,1,y,1}{C|X,X_m,Y,Y_m},
\end{align}
while assumption \eqref{jam_Y}, yields
\begin{align}
   \PboxR{a,b,c}{0,1,0,1}{A,B,C|X,X_m,Y,Y_m}  & = \delta_{b \oplus c,0} \cdot \mu \cdot \PboxR{a}{0,1,0,1}{A|X,X_m,Y,Y_m}, \label{jam_Y2a} \\
   \PboxR{a,b,c}{0,1,1,1}{A,B,C|X,X_m,Y,Y_m}  & = \delta_{b \oplus c,1} \cdot (1-\mu) \cdot \PboxR{a}{0,1,1,1}{A|X,X_m,Y,Y_m}. \label{jam_Y2b}
\end{align}
The operational no-signalling constraints \eqref{compass} together with \eqref{jam_X2} imply that Eqs. \eqref{jam_Y2a} and \eqref{jam_Y2b} must be equal for all $a,b,c$, hence
\begin{align}
 \delta_{b \oplus c,0} \cdot \mu = \delta_{b \oplus c,1} \cdot (1-\mu).
\end{align}
But this cannot hold for all $b,c$, because $b,c =0$ requires $\mu=0$, while $b =0, c =1$ necessitates $\mu=1$.

This analysis simply shows that the assumed jamming mechanism \eqref{compass} is incompatible with operational no-signalling constraints \eqref{ONS_new} in the particular ``compass'' setting of spacetime points. It does not say that ``by choosing $Y_m$ Yanina can restrict Xavier’s choice of $X_m$'', as claimed in \cite{Monogamy_HR} and it clearly does not show that \emph{any} jamming mechanism would lead to superluminal signalling. 

Observe that to arrive at the above conclusions we did not invoke the monogamy relation \eqref{eq:entropic-mono}. We now proceed to clarify that in general such relations are not directly related to the spacetime structure and no-signalling constraints.

The monogamy of the CHSH Bell inequality in no-signalling theories is well known \cite{PhysRevLett.113.210403} 
\begin{eqnarray}
\label{eq:monochsh-nosig}
    \omega(CHSH)_{AB} + \omega(CHSH)_{AC} \leq \frac{3}{2},
\end{eqnarray}
where we have considered the standard two-party CHSH expression with classical value $3/4$. That is, in general no-signalling theories (obeying the usual no-signalling constraints), the observation of a super-classical ($> \frac{3}{4}$) score by one pair of players $A,B$ precludes the observation of super-classical score by the pair $A, C$ in a configuration in which all three players are spacelike separated.

In the SI\ref{app:XORmono-signalling}, we prove a technical result showing that monogamy relations exist for general two-party correlation Bell inequalities (generalisations of the CHSH inequality) in signalling theories as well. The intuition is simple --- the winning conditions in general two-player Bell expressions are in some cases in conflict for different pairs of players conducting the same Bell experiment. So that, again even in signalling theories, for certain non-local games or two-player Bell inequalities, the observation of a super-classical score by some pairs of players precludes the observation of super-classical score by other pairs. 

In both signalling and no-signalling theories, the conclusion is the same --- the existence of such a monogamy relation cannot be taken as evidence that the theories are unphysical. To be precise, the monogamy of CHSH inequality in no-signalling theories in \eqref{eq:monochsh-nosig} is in exact analogy with the monogamy relation from \cite{Monogamy_HR} in that one may claim that what score an experimentalist observes for $A, B$  is tied to a possibly unknown space-like separated variable $C$. However, this does not make ``causal and experimental reasoning about setups that allow such correlations challenging'' \cite{Monogamy_HR}. Instead, the monogamy of CHSH inequality is a physically observed effect that lies at the very foundation of quantum cryptography.

After all, quantum theory is an evidently physical no-signalling theory with features such as entanglement that may appear to be non-local but are nevertheless very much physical. One may question the mechanisms leading to apparently nonlocal features such as entanglement but this does not on its own rule out their physicality. In conclusion, while we acknowledge the existence of monogamy relations in correlations obeying the relativistically causal constraints from \cite{PawelRaviCausality} and Thm. \ref{thm:MAIN}, there is no justification to the claim that such monogamy relations interdict the physicality of such relativistically causal theories.

\section{Discussion}

Several questions of primary importance both for fundamental physics and quantum information processing in spacetime have been unsettled: (1) The formulation of the exact constraints imposed by relativistic causality in different spacetimes \cite{PR_box, PawelRaviCausality}; (2) their implications for the security of device-independent protocols against general adversaries limited only by the no-signalling principle \cite{NS_crypto, RC_crypto}; (3) and the possible physicality of jamming of nonlocal correlations, which respect the no-superluminal-signalling constraints \cite{Monogamy_HR}.

In this paper, we have constructed a unified framework to handle both temporal and nonlocal correlations and derived the operational no-signalling constraints in general relativistic spacetimes. It is worth noting that these constraints can also be applied in other spacetime models, e.g. Newtonian spacetime \cite{Sanchez_Newton}, spacetimes with a preferred foliation \cite{Causality_foliation}, Finslerian spacetimes \cite{Minguzzi_Finsler}, causal sets \cite{CausalSet_PRL,CausalSet_Rev} or certain noncommutative spacetimes \cite{CQG2013,UNIV2017}. This is because the formulation of constraints \eqref{ONS_new} requires only suitable notions of `future' and `past', and not the full structure of a smooth spacetime manifold. On the other hand, in order to rule out causal loops and the ensuing logical paradoxes, the no-signalling constraints have to be complemented by \emph{some} assumptions about the fundamental symmetries of spacetime. In the Minkowski spacetime this is the familiar Poincar\'e symmetry, but in other spacetimes (relativistic or more general) it is not immediately clear what the right operational symmetries should be. This insight opens a new avenue in studying nonlocal correlations and their possible physical realisations.

We have applied our framework to clarify how recently provided schemes to embed causal loops in Minkowski spacetime \cite{Loops_PRL} would lead to a violation of the operational no-signalling constraints and thence to logical paradoxes, under standard assumptions on the freedom of an agent to intervene or not. We have also refuted recent arguments against the physicality of jamming of nonlocal correlations based on the monogamy relations which arise in such theories \cite{Monogamy_HR} by explicitly deriving monogamy relations even in signalling theories and showing that such monogamy relations should be interpreted as frustrations caused by opposing constraints rather than limitations on the underlying physics. In terms of fundamental physics, our findings motivate the search for possible physical mechanisms for jamming correlations.

In terms of quantum information processing, the established operational no-signalling constraints expose possible attacks by an adversary only constrained by the laws of relativity on device-independent protocols, such as for random number generation and key distribution, and reinforce the importance of implementing such protocols in carefully selected spacetime configurations. The implications for other protocols, such as for relativistic bit commitment \cite{PhysRevLett.115.030502} and quantum position verification \cite{Kavuri:25}, are interesting avenues for future research.

\section*{Acknowledgments}
 M.E., P.H and T.M. acknowledge support by the National Science Centre in Poland under the research grant Maestro (2021/42/A/ST2/00356). R.R. acknowledges support from the General Research Fund (GRF) Grant No. 17211122, and the Research Impact Fund (RIF) Grant No. R7035-21.

\section{Methods}

\subsection{Relativistic spacetime}

A relativistic spacetime $\M$ is a four-dimensional time-oriented (connected) manifold equipped with a Lorentzian metric 
(see e.g. \cite{Wald}) which, among other things, endows $\M$ with various causal relations (for a comprehensive review of causality theory, see \cite{MS08,MingRev}). 
More concretely, a point $p$ \emph{strictly causally precedes} $q$, what is denoted $p \prec q$, if there exists a piecewise smooth future-directed causal curve going from $p$ to $q$. Notice that this relation is transitive but irreflexive ($p \not\preceq p$), unlike the more commonly used causal precedence relation $p \preceq q$ defined as $p \prec q$ or $p = q$. In the current paper, however, both relations play a role (most notably in Definition \ref{opsepDef}). More abstractly, one can define the notions of future and past in any spacetime model equipped with a binary relation $\prec$.

The standard causal future of a point $p$ reads $J^+(p) = \{ q \in \M \, | \, p \preceq q\}$, whereas its \emph{strict} causal future is $R^+(p) = \{ q \in \M \, | \, p \prec q\} = J^+(p) \setminus \{p\}$. The causal past $J^-(p)$ and the strict causal past $R^-(p)$ are defined analogously.

The relation $\preceq$ is a partial order (i.e. a reflexive, antisymmetric and transitive relation) if and only if the spacetime $\M$ does not contain closed causal curves, that is there exist no points $p \neq q \in \M$ such that $p \preceq q \preceq p$. One can classify spacetimes in `ladder-like' structure according to their causal properties \cite{MS08,MingRev}, the top level being known as \emph{globally hyperbolic spacetimes}. The latter include the Minkowski spacetime, as well as i.a. Schwarzschild and FLRW spacetimes \cite{Wald}. Every spacetime is locally hyperbolic \cite{Penrose1972}.

Any experiment is performed with some physical device operating in the spacetime region $V \times T$, where $V$ is the volume of the device and $T$ its time-lapse of operation (see Fig. \ref{fig:box}). The sets $V$ and $T$ are defined in a chosen local coordinate chart $\phi: U \subset \M \to \R^4$, but the spacetime region associated with the device, $\K \vc \phi^{-1}(V \times T) \subset \M$, is a chart-independent object.

\begin{figure}[H]
\begin{center}
\resizebox{0.35\textwidth}{!}{\includegraphics[scale=1]{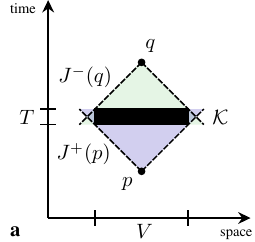}} \qquad
\caption{\label{fig:box} Spacetime diagram for a basic experiment with a single input and a single output. The `input point' $p$ marks the agent's decision on the choice of input and the `output point' $q$ marks the gathering of experimental information from the device. The cones, $J^+(p)$ and $J^-(q)$, depict the future and past of the events $p$ and $q$, respectively. Note that even if we assumed the measurement to be instantaneous, i.e. $T = \{t_0\}$, but the detector has a non-pointlike volume $V$, the points $p$ and $q$ would not coincide. 
}
\end{center}
\end{figure}

\subsection{Operational separation relation}

The operational separation relation introduced in Def. \ref{opsepDef} plays a pivotal role in the operational no-signalling constraints \eqref{ONS_new}. It could be seen as an asymmetric extension of the standard spacelike separation relation between pairs of points in a spacetime~$\M$.

Fig. \ref{fig:ocaus} provides a more detailed illustration of its features: Every single spacetime point $q_1, q_2, q_3$ and $q_4$ is operationally separated from the spacetime point $p$. Note that $p$ is actually spacelike-separated from $q_1,q_2$ and $q_3$, while $q_4$ is in the past of $p$. The pair $(q_1,q_3)$ is operationally separated from $p$, because there exists a point $Q$, which lies in the common future of $q_1$ and $q_3$, but outside of the future of $p$ (red region). Similarly, the pair $(q_2,q_4)$ is operationally separated from $p$, because there exists a point $Q'$ in the orange region of spacetime with similar properties. On the other hand, the pair $(q_1,q_2)$ is \emph{not} operationally separated from $p$, because the common future of $q_1$ and $q_2$ (blue region) is entirely contained in the green region -- the future of $p$, hence the is no spacetime point verifying the conditions in Definition \ref{opsepDef}. Also the pair $(q_2,q_3)$ is \emph{not} operationally separated from $p$, the common future of which is the violet region. The latter is also equal to the joint future of the triples $(q_1,q_2,q_3)$ and $(q_2,q_3,q_4)$, and also of the quadruple $(q_1,q_2,q_3,q_4)$. Consequently, neither $(q_1,q_2,q_3)$ nor $(q_2,q_3,q_4)$ nor $(q_1,q_2,q_3,q_4)$ is operationally separated from $p$. Finally, the triple $(q_1,q_2,q_4)$ is also \emph{not} operationally separated from $p$, as the common future of the former points coincides with the common future of $q_1$ and $q_2$ (blue region).

\begin{figure}[H]
\begin{center}
\resizebox{0.6\textwidth}{!}{\includegraphics[scale=1]{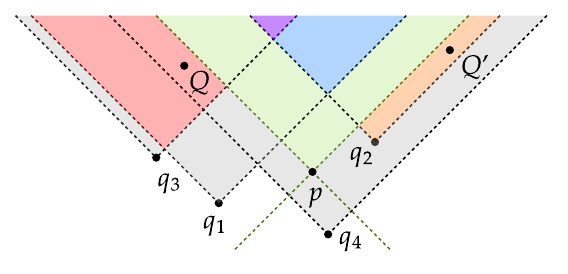}} \qquad
\caption{\label{fig:ocaus} An illustration of the operational separation relation --- see text for the description.
}
\end{center}
\end{figure}

\subsection{Spacetime random variables}

A \emph{random variable} (RV) $X$ is defined through a finite set $\X = \{x^{(1)}, \ldots, x^{(k)}\}$ and a probability distribution $P_X$ on $\X$. These have to be understood \emph{operationally}, as quantities related to a given physical system accessible to agents, and not the intrinsic properties of the system itself. For instance, the position of a particle (classical or quantum) is \emph{not} an RV, but the possible outcome of a position measurement is an RV. The values of an RV need not be numerical (though they can always be translated into numbers)  --- we could have, for instance, $\X = \{\text{red},\text{green},\text{blue}\}$ or $\X = \{\text{click},\text{no click}\}$.  We shall use the simplified notation 
$P_X(x) \vc P(X = x)$
to denote the probability that an RV $X$ acquires a value $x$.

Having fixed the spacetime structure, we promote any RV to a \emph{spacetime random variable} (SRV), following the work \cite{SRV}. An SRV is a pair $(X,p)$, where $X$ is an RV and $p$ is a point in the spacetime $\M$. 

The interplay between the (purely) operational and spacetime characteristics of SRVs induces three important properties:

Firstly, two SRVs, $(X,p)$ and $(X',q)$ associated with different points of spacetime, $p \neq q$, should be considered as different entities, even if they have the same sets of values, $\X = \X'$. Notwithstanding, we could demand $X$ and $X'$ to be perfectly correlated, through a perfect classical communication channel, provided that $q$ lies in the future of $p$.

Secondly, if we have two SRVs $(X,q_1)$ and $(Y,q_2)$, then an SRV $(f(X,Y),Q)$ is defined only for a gathering point $Q \in J^+(q_1) \cap J ^+(q_2)$ for any function $f(X,Y)$ (depending nontrivially on both arguments), cf. Fig. \ref{fig:OS}. This stems from the fact that there is no SRV if there is no agent who could read it out. 
Computing any function depending nontrivially on RVs $X,Y$ 
requires gathering information about both $X$ and $Y$ at a {\it single} spacetime point $r$.

Thirdly, the probabilities associated with SRVs are `\emph{theoretical}' probabilities of \emph{single events}. In order to obtain the `\emph{empirical}' probabilities, which are computed from a finite detection statistics, we would need multiple SRVs, from which information is gathered and processed to yield the statistics (see Supplemental Information for the details).

\subsection{Interventions and box extensions}

In the context of causal modelling \cite{Pearl} it is often assumed that an agent can ``intervene'' on a certain RV. Such an intervention consists in forcing an RV to acquire a given value. Within the developed framework the concept of an intervention is formalised as follows:

Let $(A,q)$ be an SRV with a set of values $\A = \{a^{(1)},\ldots,a^{(k)}\}$ and a probability distribution $\{P_{A}\big[ (a,q) \big]\}_{a \in \A}$. An \emph{intervention} is an SRV $(I_A,q')$ with the set of possible values $\I_A = \{\idle,\doo(a^{(1)}), \ldots, \doo(a^{(k)})\}$ and a spacetime location $q' \preceq q$. It induces a collection of conditional probability distributions $\big\{\PboxR{(a,q)}{(i_A,q')}{A|I_A}\big\}$ with the following properties:
\begin{align}\label{interv}
\PboxR{(a,q)}{(\idle,q')}{A|I_A} = P_{A}\big[ (a,q) \big], && \PboxR{(a,q)}{(\doo(a'),q')}{A|I_A} = \delta_{a,a'},
\end{align}
for all $a,a' \in \A$. Formulas \eqref{interv} are intuitively clear and consistent with the standard understanding (see e.g. \cite[Eqs. 4a--4d]{Loops_PRA}): When there is no intervention the probability of getting an outcome $a$ is the initial one, but if one forces a concrete value $a'$, one will have it with certainty.

More generally, suppose that we have a collection of SRVs $(\bA,\bq)$, $(\bX,\bp)$ together with the collection of conditional probabilities $\big\{ \PboxR{(\ba,\bq)}{(\bx,\bp)}{\bA|\bX}\big\}$. Suppose now that an agent can intervene on the SRV $(A_j,q_j)$ for some fixed $j \in \{1,\ldots, n\}$. This means that we now have an additional SRV, $(I_{A_j},q_j')$, located at $q_j'$ in the past of $q_j$. We extend the collection of conditional probabilities to
$\big\{ \PboxR{(\ba,\bq)}{(\bx,\bp),(i_{A_j},q_j')}{\bA|\bX,I_{A_j}} \big\}$ with the following properties:
\begin{align}
  & \PboxR{(\ba,\bq)}{(\bx,\bp),(\idle,q_j')}{\bA|\bX,I_{A_j}} = \PboxR{(\ba,\bq)}{(\bx,\bp)}{\bA|\bX}, \label{Ppre} \\
  & \PboxR{(\ba,\bq)}{(\bx,\bp),(\doo(a_j'),q_j')}{\bA|\bX,I_{A_j}} = \delta_{a_j,a'_j} \cdot \PboxR{(\ba,\bq)^G}{(\bx,\bp),(\doo(a_j'),q_j')}{\bA^G|\bX,I_{A_j}} , \label{Ppost}
\end{align}
where $G \vc \{1,\ldots,n\} \setminus \{j\}$. The post-intervention probabilities, on the RHS of \eqref{Ppost}, are not determined by the pre-intervention ones, that is we expect in general
\begin{align}
\PboxR{(\ba,\bq)^G}{(\bx,\bp),(\doo(a_j'),q_j')}{\mathbf{A}^G|\mathbf{X},I_{A_j}} \neq \PboxR{(\ba,\bq)^G}{(\bx,\bp),(a_j,q_j)}{\mathbf{A}^G|\mathbf{X},A_j}.
\end{align}
Along the same lines, one can introduce interventions on all other SRVs in the setup.

\medskip

Three important comments are in order:

\medskip

Firstly, a collection of conditional probabilities \eqref{box} is general and includes the possibility that some of the SRVs $(X_i,p_i)$ are actually interventions. This can be seen in two ways: If we assume that the intervention on an SRV $(A_j,q_j)$ is introduced by the same agent who selects a value $x_j$ of the input SRV $(X_j,p_j)$, then we can set $q_j = p_j$ and promote the SRV $(X_j,p_j)$ to $((X_j,I_{A_j}),p_j)$ which takes values from $\X_j \times \I_{A_j}$. Alternatively, if we assume that the intervention is carried out by an independent agent, we can define an additional SRV as $(X_{n+1},p_{n+1}) \vc (I_{A_j},q_j')$ and trivially extend the setting by defining a corresponding fictitious SRV $(A_{n+1},q_{n+1})$ with a void set of values, $\A_{n+1} = \emptyset$.

Secondly, since our framework is purely operational, we do not specify a priori any (classical or quantum) causal model, which determines the pre- and post-intervention probabilities, and actually we do not even assume that such a model must exist. This implies, in particular, that all setups based on probabilities determined by pre- and post-intervention causal structures, such as \cite{Ringbauer16,Loops_PRA,Loops_PRL,Grothus24,ColbeckVilasini25} etc. 
are included in our framework.

Thirdly, an intervention is a (special kind of) an input RV, so it can always be used by the agent to encode information into the (possibly nonlocal) system at hand. Consequently, the extended set of probabilities --- whether involving an extended input or a new agent --- must abide by the no-signalling constraints \eqref{ONS_new}. In other words, the agent performing the intervention must not be able to use the intervention to (statistically) transmit any information outside of his future. In particular, he must not be able to signal through the choice of whether to intervene or not.

\section{Supplementary Information}

\subsection{Empirical probabilities from SRVs}\label{subsec:SRVs}

A single SRV $(X,p)$ is associated with theoretical probabilities of a single event. In order to consider the empirical probabilities we need $N$ SRVs,  $\{(X_k,p_k)\}_{k=1}^N$ with the same set of possible values, $\X_k = \X$ for all $k$, but different points, $p_k \neq p_l$ for $k \neq l$. 

We have three possible measurement schemes:
\begin{enumerate}
    \item An array of detectors working simultaneously. This corresponds to $p_j \nprec p_k$ for all $j \neq k$.
    \item Consecutive measurement with the same detector. This corresponds to $p_j \prec p_k$ for $j < k$.
    \item A combination of 1. and 2.
\end{enumerate}

Now, we define an SRV $(X,Q)$ with the set of values $\X$, which simply counts the number of occurrences of a given value $x$ among the outcomes $\{x_k\}_{k=1}^N$. It has a \emph{frequency} probability distribution $P(x) \vc \tfrac{1}{N} \sum_k \delta_{x_k,x}$. Whichever way of building the statistics we choose, we must eventually \emph{gather} the information from all detectors. This requires that $Q \in \cap_{k=1}^N J^+(p_k)$.

A natural generalisation of this setting is to associate SRVs with \emph{regions} of spacetime, rather than points. An SRV $(X,\K)$ corresponds to a \emph{potential} measurement statistics of an RV $X$ gathered in a spacetime region $\K \subset \M$. This will be covered in a forthcoming paper.

\subsection{Properties of the operational separation relation}
\label{subsec::bipartite}

Let us recall the definition of operational separation of tuples of spacetime points, presented in the main text, in a slightly more general form.

\begin{Def}
\label{opsepDef_rs}
We say that an $s$-tuple $\bq$ of spacetime points is \emph{operationally separated} from an $r$-tuple $\bp$ if there exists a point $Q \in \M$, such that $q_j \preceq Q$ and $p_i \nprec Q$ for all $i \in \{1,\ldots,r\}$, $j \in \{1,\ldots,s\}$.
\end{Def}

This condition 
can be equivalently expressed as
\begin{align*}
\exists Q \in \M \quad \bq \subset J^-(Q) \quad \textnormal{and} \quad Q \not\in R^+(\bp).
\end{align*}

The complementary condition that $\bq$ is \emph{not} operationally separated from $\bp$ defines the following relation between finite subsets of~$\M$
\begin{align}
\label{rel1}
\bp \prec_o \bq \quad \Leftrightarrow \quad \left[ \forall Q \in \M \quad \bq \subset J^-(Q) \ \Rightarrow \ Q \in R^+(\bp)\right].
\end{align}
This relation can be equivalently defined without referring to a gathering point $Q$ by the inclusion
\begin{align}
\label{rel2}
  \bp \prec_o \bq \quad \Leftrightarrow \quad \bigcap\nolimits_{j} J^+(q_j) \subset \bigcup\nolimits_{i} R^+(p_i).
\end{align}
Note that the intersection $\cap_{j} J^+(q_j)$ is exactly the set of all gathering points --- points at which it is possible to gather the information from all $q_j$'s.

In the case when $s=1$, for any fixed $r$, the relation $\prec_o$ reduces to the following condition
\begin{align}\label{rel3}
\bp \prec_o q \quad \Leftrightarrow \quad \exists i \ \ p_i \prec q.
\end{align}
Indeed, the latter can be equivalently expressed as $q \in \cup_i R^+(p_i)$, which, by the transitivity of $\prec$, is in turn equivalent to the inclusion $J^+(q) \subset \cup_i R^+(p_i)$. But this is nothing but $\bp \prec_o q$ on the strength of characterisation \eqref{rel2}.

Thus, for $s=1$ the $\prec_o$ relation reduces to the standard strict causal relation $\prec$. This was to be expected --- since in this case we are interested in the outcomes in a single spacetime point $q$, all the relevant information is already gathered in that point, ready to be operationally accessed.

Therefore, in the case $s=1$, the NS constraints \eqref{ONS_new} can be expresses simply as 
\begin{align}
\text{If } \quad p_i \not\prec q, \text{ for all } i \in F, \quad \text{ then } \quad \PboxR{(a,q)}{(\bx,\bp)}{A|\bX} = \PboxR{(a,q)}{(\bx',\bp)}{A|\bX}
\end{align}
for any $a$, any $\bx$ and any $\bx'$, such that $x'_i = x_i$, for $i \notin F$. Any violation of this condition could in principle be utilised to perform the most basic form of signalling --- see \cite{PRA2020} for the in-depth analysis of the $r=s=1$ case.

Observe, finally, that one has the following implications
\begin{align}\label{prop1}
\exists i,j \ \ p_i \prec q_j \ \ \begin{array}{lcr} \rotatebox[origin=c]{15}{$\Rightarrow$} & \ \exists i \ \ p_i \prec_o \bq & \rotatebox[origin=c]{-15}{$\Rightarrow$}
\\[4pt] \rotatebox[origin=c]{-15}{$\Rightarrow$} & \ \exists j \ \ \bp \prec_o q_j & \rotatebox[origin=c]{15}{$\Rightarrow$}
\end{array}
\ \ \bp \prec_o \bq.
\end{align}
Indeed, the first condition can be equivalently expressed as $J^+(q_j) \subset R^+(p_i)$ for some $i,j$, what immediately implies both that $\bigcap_{j} J^+(q_j) \subset R^+(p_i)$ for some $i$, as well as that $J^+(q_j) \subset \bigcup_{i} R^+(p_i)$ for some $j$. On the strength of \eqref{rel2}, these conditions are equivalent to $\exists i \ p_i \prec_o \bq$ and to $\exists j \ \bp \prec_o q_j$, respectively. Furthermore, each of the above inclusions in turn implies the inclusion $\bigcap_{j} J^+(q_j) \subset \bigcup_{i} R^+(p_i)$, equivalent to the final condition $\bp \prec_o \bq$.

Equivalently, one can rephrase \eqref{prop1} as saying that if $\bq$ is operationally separated from the tuple $\bp$ then
\begin{itemize}
\item $\bq$ is operationally separated from every single point $p_i$ for every $i \in \{1,\ldots,r\}.$
\item For every $j \in \{1,\ldots,s\}$, $q_j$ is operationally separated from $\bp$.
\item Each of the above two conditions further implies that $p_i \not\prec q_j$ for all $i,j$.
\end{itemize}

\subsection{Proof of Theorem \ref{thm:MAIN}}

We first show that if operational superluminal signalling is possible in a physical theory yielding a set of experimental conditional probabilities \eqref{box} then the no-signalling constraints \eqref{ONS_new} are violated. Indeed, the possibility of operational signalling means --- by definition --- that there exists an input SRV $(X,p)$ and an output SRV $(B,q)$ accessible to two agents, such that there exist $x \neq x' \in \X$ and $b \in \B$ such that
\begin{align}
    \PboxR{(b,q)}{(x,p)}{B|X} \neq \PboxR{(b,q)}{(x',p)}{B|X}.
\end{align}
The signalling is superluminal if the readout point $q$ lies outside of the future light cone of the sending point $p$, that is $p \nprec q$. On the strength of formula \eqref{rel3} the point $q$ is operationally separated from $p$ if and only if $p \nprec q$, and hence condition \eqref{ONS_new} is violated.

The proof that the constraints \eqref{ONS_new} are also necessary for the lack of operational superluminal signalling requires a bit more care. Suppose \eqref{ONS_new} is violated. This means that there exist $n$-tuples of SRVs $(\bA,\bq)$, $(\bX,\bp)$ and subsets of indices $F,G \subset \{1,\ldots,n\}$, such that $\bq^G$ is operationally separated from $\bp^F$, but
\begin{align}\label{ass}
    \PboxR{(\ba,\bq)^G}{(\bx,\bp)}{\bA^G|\bX} \neq \PboxR{(\ba,\bq)^G}{(\bx',\bp)}{\bA^G|\bX},
\end{align}
for some $\ba$, $\bx$, $\bx'$ with $x_i' = x_i$ for all $i \notin F$. Let us now express $F$ explicitly as $\{i_1,\ldots,i_r\}$ and for any $H \subset \{1,\ldots,n\}$ denote its complement as $H^\complement \vc \{1,\ldots,n\} \setminus H$. Consider the following list of $(r+1)$ conditional probabilities, where each two consecutive items differ on at most one setting:
\begin{align*}
    & \PboxR{(\ba,\bq)^G}{(\bx,\bp)}{\bA^G|\bX}
    \\
    & \PboxR{(\ba,\bq)^G}{(\bx',\bp)^{\{i_1\}}, (\bx,\bp)^{\{i_1\}^\complement}}{\bA^G|\bX}
    \\
    & \PboxR{(\ba,\bq)^G}{(\bx',\bp)^{\{i_1,i_2\}}, (\bx,\bp)^{\{i_1,i_2\}^\complement}}{\bA^G|\bX}
    \\
    & \qquad \vdots
    \\
    & \PboxR{(\ba,\bq)^G}{(\bx',\bp)^{F}, (\bx,\bp)^{F^\complement}}{\bA^G|\bX}.
\end{align*}

By assumption \eqref{ass}, the first and the last items in the above list are \emph{not} equal. Hence, there exists at least one pair of consecutive items which are not equal, i.e., there exists $k$ such that
\begin{align*}
   & \PboxR{(\ba,\bq)^G}{(\bx',\bp)^{\{i_1,\ldots,i_{k-1}\}} , (x_{i_k}, p_{i_k}),  (\bx,\bp)^{\{i_1,\ldots,i_k\}^\complement}}{\bA^G|\bX}\\
\neq\, & \PboxR{(\ba,\bq)^G}{(\bx',\bp)^{\{i_1,\ldots,i_{k-1}\}} , (x'_{i_k}, p_{i_k}),  (\bx,\bp)^{\{i_1,\ldots,i_k\}^\complement}}{\bA^G|\bX}.
\end{align*}
Observe now that because $\bq^G$ is operationally separated from $\bp^F$, then by the top-right implication in \eqref{prop1} $\bq^G$ is operationally separated from $p_{i_k}$. This means that there exists a gathering event $Q$ for $\bq^G$ such that $p_{i_k} \nprec Q$. Consequently, an agent who can freely choose between the input values $x_{i_k}$ and $x_{i_k}'$ of the SRV $(X_{i_k},p_{i_k})$ is able to statistically communicate a bit to an agent who gathers the information from the SRVs $(\bA,\bq)^G$ at the point $Q$. This constitutes an operational protocol for statistical superluminal signalling.

\subsection{General `Bell-type' scenarios}

A standard Bell scenario \cite{Bell_Nonlocal} concerns $n$ spacelike separated agents performing local measurements, with $m$ settings and $k$ outputs, on a joint physical system. It includes 2$n$ random variables, $X_1,\ldots,X_n$ and $A_1,\ldots,A_n$ with values $x_i \in \{1,\ldots,m\}$ and $a_i \in \{1,\ldots,k\}$ and a set of conditional probabilities:
\begin{align}\label{box3}
P\big(a_1, \ldots, a_n \, \big\vert \, x_1, \ldots, x_n \big).
\end{align}

Typically, it is assumed that $X_i$ and $A_i$ are associated with same spacetime point $p_i$ and that %
the SRVs are mutually spacelike separated, $p_i \npreceq p_j$ for all $i \neq j$. For such a spacetime configuration one assumes the `no-signalling constraints':
\begin{align}\label{NS_old}
 \sum_{a_j} P(a_1,\ldots,a_j,\ldots,a_n \, | \, x_1,\ldots,x_j,\ldots,x_n) = \sum_{a_j} P(a_1,\ldots,a_j,\ldots,a_n \, | \, x_1,\ldots,x_j',\ldots,x_n), 
\end{align}
for any fixed $j \in \{1, \ldots, n\}$ and all possible values of $a_1,\ldots,a_{j-1},a_{j+1},\ldots,a_n$,  $x_1,\ldots,x_n$ and $x_j'$. One can show that it implies the constraints for all marginal distributions of smaller-sized subsets.

These conditions assure basic compatibility with relativity --- the spacelike separated parties cannot communicate (even statistically).  The conditions \eqref{NS_old} are sufficient for the lack of operational superluminal information transfer. However, as shown in \cite{PawelRaviCausality}, they are not necessary for $n \geq 3$ parties.

The most general `Bell-type' scenario, $B\big(n,(\X_i)_{i=1}^n, (\A_i)_{i=1}^n\big)$, for unveiling correlations between $n$ parties  is defined as follows: For all $i \in \{1,\ldots,n\}$ the sets $\X_i$ and $\A_i$ determine the sets of possible values the random variables $X_i$ and $A_i$, respectively, can take. These sets can have different natures for different variables. For instance, we could have $\X_1 = \{0,\pi/4,\pi/2,3\pi/4,\pi\}$ --- parametrising the settings of a polariser,  $\X_2 = \{1\,\textup{T},2\,\textup{T},3\,\textup{T}\}$ --- the strength of the applied magnetic field or, more abstractly, $\X_3 = \{\text{no intervention, intervention}\}$. Similarly, different $\A_i$'s could encode the possible values of different physical quantities, such as position, energy, time of arrival or, simply, a detector click or the lack thereof.

In principle, the sets of values $\X_i$, $\A_i$ could be continuous, but in practice in any experiment both the inputs and the outputs are binned, because of the finite precision and resolution of any physical devices. Consequently, we can assume without loss of generality that all $\X_i$ and $\A_i$ are discrete and finite. Then, any such general Bell-type scenario can be embedded in the standard Bell scenario $B(n,m,k)$ for $m = \max_i \# \X_i$ and $k = \max_i \# \A_i$ large enough. Indeed, let us set $m_i = \# \X_i$ (with $\# \{\emptyset\} = 1$),   take some ordering functions $f_i: \{1,\ldots,m_i\} \to \X_i$ and define new random variables $Y_i$ with values in $\{1,\ldots,m_i\}$. Similarly, for the outputs we chose some $g_i$'s and define new RV's $B_i$ with values in $\{1,\ldots,k_i\}$. Then, given the set of probabilities of the original Bell-type scenario we define
\begin{align}
P(b_1,\ldots,b_n \, | \, y_1, \ldots, y_n ) = \begin{cases} 0,\qquad\qquad \qquad  \text{ if either } y_i > m_i \text{ or } b_i > k_i \text{ for some }i,
\\
P\big( g_1(b_1),\ldots,g_n(b_n) \, \big| \, f_1(y_1), \ldots, f_n(y_n) \big), \quad \text{ otherwise}.
\end{cases}
\end{align}  

\subsection{Ruling out closed causal loops in Minkowski spacetime using operational no-signalling conditions}

In \cite{Loops_PRL}, the authors had an interesting claim. Namely, it was shown that if one only forbids superluminal signalling while allowing for the possibility of superluminal causal influences, then there arises the mathematical possibility of closed causal loops embeddable in $(1+1)$-Minkowski spacetime that nevertheless do not lead to superluminal signalling. This was done by means of a general framework for causality developed by the authors in \cite{Loops_PRA}, and by a concrete construction of an ``operationally detectable causal loop'' embeddable in a $(1+1)$-Minkowski spacetime based on a fine-tuned classical causal model.

Here, we reexamine this claim in light of the operational no-signalling conditions of Def. \ref{opcausDef} which were shown to be necessary and sufficient to rule out operational signalling in Thm. \ref{thm:MAIN}. 

The authors of \cite{Loops_PRL} considered three causal models that they termed (i) the one-time pad model $\mathcal{G}^\textup{OTP}$, (ii) the jamming model $\mathcal{G}^\textup{jam}$ and (iii) the closed causal loop model $\mathcal{G}^\textup{loop}$. While the three models differed in their causal structures represented in terms of a Directed Acyclic Graph (DAG), they all lead to the same probability distribution on the observed random variables. Specifically, all three models had the same observed random variables $A, B, C$ (which are taken to be classical and binary) with the same probability distribution:
\begin{align}\label{boxL}
    P_{ABC}(a,b,c) = \begin{cases}
        \tfrac{1}{4}, & \text{ if } b= a \oplus c,\\
        0, & \text{ otherwise}.
    \end{cases}
\end{align}
Formulating these as SRVs $(A, q_1), (B, q_2)$ and $(C, q_3)$ a further commonality is that in all three cases $q_1 \npreceq q_3$ and $q_3 \npreceq q_1$. 

In $\mathcal{G}^\textup{OTP}$, the spacetime configuration is such that $J^+(q_2) \subseteq J^+(q_1) \cap J^+(q_3)$, with the causal mechanism being $A = E_A, C = E_C, B = A \oplus C$ with $E_A, E_C$ being independent and uniformly distributed. In this case, $B$ is uniformly distributed with $P_{\mathcal{G}^\textup{OTP}}(B|AC) \neq P_{\mathcal{G}^\textup{OTP}}(B)$, $P_{\mathcal{G}^\textup{OTP}}(B|A) = P_{\mathcal{G}^\textup{OTP}}(B)$ and $P_{\mathcal{G}^\textup{OTP}}(B|C) = P_{\mathcal{G}^\textup{OTP}}(B)$. In other words, $\{A,C\}$ jointly affects $B$ while neither $A$ nor $C$ individually affects $B$.

In $\mathcal{G}^\textup{jam}$, the spacetime configuration is such that $J^+(q_1) \cap J^+(q_3) \subseteq J^+(q_2)$, with the causal mechanism being $A = \Lambda = E_{\Lambda}, B = E_B, C = B \oplus \Lambda$ for an unobserved variable $\Lambda$ with $E_{\Lambda}, E_B$ being independent and uniformly distributed. In this case, $P_{\mathcal{G}^\textup{jam}}(AC|B) \neq P_{\mathcal{G}^\textup{jam}}(AC)$, $P_{\mathcal{G}^\textup{jam}}(C|B) = P_{\mathcal{G}^\textup{jam}}(C)$ and $P_{\mathcal{G}^\textup{jam}}(A|B) = P_{\mathcal{G}^\textup{jam}}(A)$. In other words, $B$ affects $\{A, C\}$ while $B$ does not affect either $A$ or $C$ individually. 

In the fine-tuned causal loop $\mathcal{G}^\textup{loop}$, the fine-tuned spacetime configuration is such that  $J^+(q_1) \cap J^+(q_3) = J^+(q_2)$, with the causal mechanism being $A = \Lambda = E_{\Lambda}$, $C = B \oplus \Lambda$, $B = A \oplus C$ with $E_{\Lambda}$ being uniformly distributed. In this case, the post-intervention causal structure $\mathcal{G}^\textup{loop}_{\textup{do}(AC)}$ is identical to $\mathcal{G}^\textup{OTP}$ and the post-intervention causal structure $\mathcal{G}^\textup{loop}_{\textup{do}(B)}$ is identical to $\mathcal{G}^\textup{jam}$. That is in this case we have that $\{A,C\}$ jointly affects $B$ \textit{and} $B$ affects $\{A, C\}$ while there are no pairwise affects relations between $A, B, C$. Furthermore, by considering an experiment in which one runs through all possible interventions on $A, C$ and another experiment in which all possible interventions are performed on $B$, one could ``operationally detect'' the closed causal loop. In other words, here we have a situation in which no superluminal signalling is taking place (all observable signalling is to the causal future), yet one has a closed causal loop embedded in $(1+1)$-Minkowski spacetime. This leads to an apparent contradiction with earlier claims in \cite{PawelRaviCausality} where similar conditions to those in Thm. \ref{thm:MAIN} were claimed to be necessary and sufficient to rule out closed causal loops. Let us now elucidate how the improved Thm. \ref{thm:MAIN} helps to rule out operational signalling and causal loops that arise as a result of operational signalling. 

In the framework considered in this paper, the interventions amount to ``inputs'' $(I_A, p_1), (I_B, p_2)$ and $(I_C, p_3)$ to a box with $p_i = q_i$. The input $I_A$ can take values `$\textup{idle}$' (corresponding to no intervention) and `$\textup{do}(a')$' corresponding to setting the value of random variable $A$ to $a'$, and similarly for the other inputs $I_B$ and $I_C$. Furthermore, the authors assume that interventions can be made on all three SRVs. Focusing on the model $\mathcal{G}^\textup{loop}$, we see that this translates in our framework to an extended box 
\begin{align}\label{box_loops}
    \PboxR{(a,q_1), (b,q_2), (c,q_3)}{(i_A,p_1), (i_B,p_2), (i_C,p_3)}{ABC|I_A I_B I_C}.
\end{align}
The box \eqref{box_loops} is subject to the following relations (for sake of brevity we omit the RVs' names and spacetime locations, as well as use a simplified notation for the marginals, e.g., $P(b) = \sum_{a,c} P(a,b,c)$):
\begin{align}
   & P\big( a,b,c \, \big\vert \, \idle,\idle,\idle \big) = P(a,b,c), \label{Pl1}\\
   & P\big( a,b,c \, \big\vert \, \doo(a'),\doo(b'),\doo(c') \big) = \delta_{a,a'} \cdot \delta_{b,b'} \cdot \delta_{c,c'},\label{Pl2}\\
   & P\big( a,b,c \, \big\vert \, \doo(a'),\idle,\doo(c') \big) = \delta_{a,a'}  \cdot \delta_{c,c'} \cdot P\big(b \, \big\vert \, \doo(a),\doo(c) \big), \label{Pl3}\\
   & P\big( a,b,c \, \big\vert \, \idle,\doo(b'),\idle \big) = \delta_{b,b'} \cdot P\big(a,c \, \big\vert \, \doo(b) \big), \label{Pl4}\\
   & P\big( a,b,c \, \big\vert \, \doo(a'),\idle,\idle \big) = \delta_{a,a'} \cdot P\big(b,c \, \big\vert \, \doo(a) \big), \label{Pl5}\\
   & P\big( a,b,c \, \big\vert \, \idle,\idle,\doo(c') \big) = \delta_{c,c'} \cdot P\big(a,b \, \big\vert \, \doo(c) \big), \label{Pl6}\\
   & P\big( a,b,c \, \big\vert \, \doo(a'),\doo(b'),\idle \big) = \delta_{a,a'} \cdot \delta_{b,b'} \cdot P\big(c \, \big\vert \, \doo(a), \doo(b) \big), \label{Pl7}\\
   & P\big( a,b,c \, \big\vert \, \idle,\doo(b'),\doo(c') \big) = \delta_{c,c'} \cdot \delta_{b,b'} \cdot P\big(a \, \big\vert \, \doo(b), \doo(c) \big). \label{Pl8}
\end{align}
The first two conditions, \eqref{Pl1} and \eqref{Pl2}, are fixed by the consistency constraints \eqref{interv}. The constraints that $\{A,C\}$ jointly affects $B$ and at the same time $B$ affects $\{A, C\}$ while no pairwise affects relations exist between $A, B, C$ translate to the following demands on these probabilities: 
\begin{align}
    & \forall\, a,b,c \quad \Pbox{a}{\doo(b)} = \Pbox{a}{\doo(c)} = P(a),\\
    & \forall\, a,b,c \quad \Pbox{b}{\doo(a)} = \Pbox{b}{\doo(c)} = P(b),\\
    & \forall\, a,b,c \quad \Pbox{c}{\doo(a)} = \Pbox{c}{\doo(b)} = P(c),\\
    & \exists\, a,b,c \quad \Pbox{b}{\doo(a),\doo(c)} \neq P(b), \label{jam_viol}\\
    & \exists\, a,b,c \quad \Pbox{a,c}{\doo(b)} \neq P(a,c).
\end{align}
As we have mentioned, the spacetime configuration of $\mathcal{G}^\textup{loop}$ is fine-tuned to obey $J^+(q_1) \cap J^+(q_3) = J^+(q_2)$, i.e., the SRV $B$ is located exactly at the intersection of the causal futures of $A$ and $C$. Firstly, note that within the relativistic model of a spacetime, such a configuration can only be realized with a single spatial dimension, in which case any spacetime manifold is conformally flat --- that is, it has the same relativistic causal structure as ($1+1$)-Minkowski spacetime. 

Secondly, with the definition of strict future detailed earlier, we observe that $q_2$ is \emph{not} in the strict future of $p_2$ (because they are the same point!). In this case, we see that $\{q_1,q_3\}$ is operationally separated from $p_2$. Indeed, there exists a point $Q = p_2$, at which the information on both $A$ and $C$ is gathered, so $q_1,q_3 \prec Q$, while $p_2 \nprec Q$. Now we see that the operationally causal constraints in \eqref{ONS_new} from Def. \ref{opcausDef} and Thm. \ref{thm:MAIN} are violated by the probability distributions forming the above box. Specifically, writing \eqref{jam_viol} out explicitly, we see that there exists at least one set of Alice's input $\doo(a')$, Charlies's input $\doo(c')$ and Bob's output $b$ such that
\begin{align}
\Pbox{b}{\doo(a'),\idle,\doo(c')} \neq \Pbox{b}{\doo(a'),\idle,\idle} = \Pbox{b}{\idle,\idle,\doo(c')} = \Pbox{b}{\idle,\idle,\idle}.
\end{align}
In other words, Bob's marginal probability distribution depends on Alice's and/or Charlie's input. We therefore conclude that \eqref{jam_viol} violates the operational no-signalling constraints \eqref{ONS_new}, since $\{q_1,q_3\}$ is operationally separated from $p_2$.

\subsection{Monogamy of XOR games in signalling theories}\label{app:XORmono-signalling}


In \cite{Monogamy_HR}, certain monogamy relations for nonlocal correlations in jamming scenarios were derived. On the basis of these monogamy relations, it was claimed that jamming is unphysical - in particular, it was claimed that any physical mechanism for jamming would lead to superluminal signalling. 

In this section, we refute the claim by explicitly proving the existence of non-trivial monogamy relations even in signalling theories (in evidently physical configurations where players signal to each other via direct communication at the speed of light). We clarify that the existence of such nontrivial monogamies should be interpreted as frustrations caused by opposing constraints in the Bell inequalities (or equivalently in the nonlocal games) rather than limitations of the underlying physics. As such these monogamies can not be used to disprove the physicality of jamming, just as evidently as they can not be used to disprove the physicality of signalling.

For convenience, in this section we work with two-player non-local games which are equivalent to bipartite Bell inequalities. The value of a general two-player non-local game $G_{AB}$ is written as
\begin{equation}
\omega(G_{AB}) := \sum_{x,y,a,b} \mu(x,y) V(a,b,x,y) P(a,b|x,y),
\end{equation}
where $x \in X, y \in Y$ denote the inputs of the two players Alice and Bob respectively, while $a,b$ denote their respective outputs. The probability distribution with which the inputs are chosen (by a referee in a non-local game) is denoted by $\mu(x,y)$ with $0 \leq \mu(x,y) \leq 1$ - when the distribution is of product form $\mu^A(x) \cdot \mu^B(y)$, the game is said to be free. The predicate $V(a,b,x,y) \in \{0,1\}$ defines the winning condition for the game, and the set of conditional distributions $\{P(a,b|x,y)\}$ is considered to belong to some well-defined theory, such as classical $\mathcal{C}$, quantum $\mathcal{Q}$, no-signalling $\mathcal{NS}$ and signalling $\mathcal{S}$ theories. In general, $\mathcal{C} \subseteq \mathcal{Q} \subseteq \mathcal{NS} \subseteq \mathcal{S}$. 

In this note, we are interested in general signalling theories, wherein the probability distributions $P(a,b|x,y)$ are only constrained by the conditions of non-negativity ($P(a,b|x,y) \geq 0$ for all $a,b,x,y$) and normalization ($\sum_{a,b} P(a,b|x,y) = 1$ for all $x,y$). We will focus on two-player XOR games also known as bipartite correlation Bell inequalities \cite{PhysRevLett.113.210403} wherein the players' outputs are binary ($a,b \in \{0,1\}$) and the winning condition depends only on the XOR of the players' outputs. In other words, the bipartite XOR game is defined by a predicate of the form $V(a \oplus b = f(x,y))$ with a function of the inputs $f(x,y) \in \{0,1\}$ and where $\oplus$ denotes the XOR. 

Furthermore, we will focus on games where the input alphabets of the two players are the same, $|X| = |Y| = m$ for $m \geq 2$, the function $f$ is defined for all input pairs $(x,y)$ (the function $f$ is a total function) and where the input distribution is uniform, i.e., $\mu(x,y) = \frac{1}{m^2}$ for all $x,y$. As we shall see, our findings will admit ready generalization to games where the above conditions are not met. 

We aim to find the monogamy relation between the values of such games achievable by pairs of players $(A,B)$, $(A,C)$ and $(B,C)$ in any signalling theory. That is, we aim to calculate
\begin{equation}
\max_{\{P(a,b,c|x,y,z)\} \in \mathcal{S}} \bigg[\omega(G_{AB}) + \omega(G_{AC}) + \omega(G_{BC}) \bigg],
\end{equation} 
where $G$ refers to a generic two-player XOR game of the above type, the maximisation is over all tripartite signalling boxes $\{P(a,b,c|x,y,z)\}$ of three players with inputs $x,y,z$ and outputs $a,b,c$ respectively. At this point, one has to be careful since a value such as $\omega(G_{AB})$ could depend on the specific value of the input $z = 0$ or $z=1$ of the player $C$. At this point, we focus on a definition of $\omega(G_{AB})$ as the winning probability of Alice and Bob averaged over all inputs $z$ of Charlie. We prove the following theorem showing the existence of nontrivial monogamy relations for general bipartite XOR games in even signalling theories!
\begin{Thm}
    Let $G$ refer to a total function two-player XOR game with $m$ inputs per player, and let the input distribution be uniform. Denote by $\mathcal{S}$ the set of tripartite signalling behaviors. Then the maximum of the sum of pairwise winning probabilities achievable by pairs of players $(A,B)$, $(A,C)$ and $(B,C)$ in any signalling theory obeys
    \begin{eqnarray}
\max_{\{P(a,b,c|x,y,z)\} \in \mathcal{S}} \bigg[\omega(G_{AB}) + \omega(G_{AC}) + \omega(G_{BC}) \bigg] = \frac{2}{m^3} \left(|S_{AAA}| + |S_{ACC}|\right) + \frac{3}{m^3} \left(|S_{CCC}| + |S_{AAC}| \right),
\end{eqnarray}
where 
\begin{eqnarray}
S_{AAA} &:=& \left\{ (x,y,z) \in X \times Y \times Z \big| f(x,y) = f(y,z) = f(x,z) = 1 \right\}, \nonumber \\ 
S_{ACC} &:=& \left\{ (x,y,z) \in X \times Y \times Z \big| (f(x,y), f(y,z), f(x,z)) \in \{ (1,0,0), (0,1,0), (0,0,1) \} \right\}, \nonumber \\ 
S_{AAC} &:=& \left\{ (x,y,z) \in X \times Y \times Z \big| (f(x,y), f(y,z), f(x,z)) \in \{ (1,1,0), (1,0,1), (0,1,1) \} \right\}, \nonumber \\ 
S_{CCC} &:=& \left\{ (x,y,z) \in X \times Y \times Z \big| f(x,y) = f(y,z) = f(x,z) = 0 \right\}.
\end{eqnarray}
\end{Thm}

\begin{proof}
We have
\begin{eqnarray}
\label{eq:mono-S}
&&\max_{\{P(a,b,c|x,y,z)\} \in \mathcal{S}} \bigg[\omega(G_{AB}) + \omega(G_{AC}) + \omega(G_{BC}) \bigg] \nonumber \\
&=& \max_{\{P(a,b,c|x,y,z)\} \in \mathcal{S}} \bigg[ \sum_{x,y} \frac{1}{m^2} V(a\oplus b = f(x,y)) \sum_{z} \frac{1}{m} \sum_{c} P(a,b,c|x,y,z)  \nonumber \\
&& \qquad  \qquad \qquad \qquad + \sum_{y,z} \frac{1}{m^2} V(b\oplus c = f(y,z)) \sum_{x} \frac{1}{m} \sum_{a} P(a,b,c|x,y,z) \nonumber \\
&& \qquad \qquad \qquad \qquad + \sum_{x,z} \frac{1}{m^2} V(a\oplus c = f(x,z)) \sum_{y} \frac{1}{m} \sum_{b} P(a,b,c|x,y,z) \bigg] \nonumber \\
&=& \max_{\{P(a,b,c|x,y,z)\} \in \mathcal{S}} \frac{1}{m^3} \sum_{x,y,z} \bigg[ \sum_{c} V(a \oplus b = f(x,y)) + \sum_{a} V(b \oplus c = f(y,z)) + \sum_{b} V(a \oplus c = f(x,z)) \bigg] P(a,b,c|x,y,z). \nonumber \\
\end{eqnarray}
Now, we note that the set of signalling correlations $\mathcal{S}$ is the convex hull of deterministic behaviors $P(a,b,c|x,y,z) \in \{0,1\}$ with $P(a,b,c|x,y,z)$ obeying the constraints of non-negativity and normalisation. As such, the maximization in \eqref{eq:mono-S} is achieved at a deterministic point, and the maximization can be thought of as an Integer Linear Program (the objective function and the constraints are linear, while the variables $P(a,b,c|x,y,z)$ take integer values in $\{0,1\}$). While an Integer Linear Program is in general hard to solve, crucially in the case of Signalling theories, we observe that the maximization can be done individually term-by-term. This is because the global maximum in this case can be achieved as a sum over all the maximum values of individual terms. This observation leads us to consider the maximum value of the individual terms in \eqref{eq:mono-S} for each triple $(x,y,z)$. We observe the following:
\begin{fact}
\label{fact:frustrate-XOR}
For all input triples $(x,y,z)$ such that $\left( f(x,y), f(y,z), f(x,z) \right) \in \big\{(0,0,0), (0,1,1), (1,0,1), (1,1,0) \big\}$ there exists a normalised distribution $P(a,b,c|x,y,z)$ such that 
\begin{eqnarray}
\sum_{c} V(a \oplus b = f(x,y)) P(a,b,c|x,y,z) = \sum_{a} V(b \oplus c = f(y,z)) P(a,b,c|x,y,z) = \sum_{b} V(a \oplus c = f(x,z)) P(a,b,c|x,y,z) = 1. \nonumber \\
\end{eqnarray}
On the other hand, for input triples $(x,y,z)$ such that $\left( f(x,y), f(y,z), f(x,z) \right) \in \big\{(0,0,1), (0,1,0), (1,0,0), (1,1,1) \big\}$, we have
\begin{equation}
\max_{P(a,b,c|x,y,z)}  \bigg[ \sum_{c} V(a \oplus b = f(x,y)) + \sum_{a} V(b \oplus c = f(y,z)) + \sum_{b} V(a \oplus c = f(x,z)) \bigg] P(a,b,c|x,y,z) = 2.
\end{equation}
\end{fact}
The fact follows by considering that the cases $\left( f(x,y), f(y,z), f(x,z) \right) \in \big\{(0,0,0), (0,1,1), (1,0,1), (1,1,0) \big\}$ are \textit{frustration-free}, i.e., there is an output triple $(a,b,c) \in \{0,1\}^3$ that can satisfy the winning conditions of all three games in this case. On the other hand, for input triples such that $\left( f(x,y), f(y,z), f(x,z) \right) \in \big\{(0,0,1), (0,1,0), (1,0,0), (1,1,1) \big\}$ no output triple $(a,b,c) \in \{0,1\}^3$ exists that can satisfy all three winning conditions simultaneously even when the inputs of all players are broadcast to all parties, and at most two winning conditions can be simultaneously satisfied. 

Using the Fact \ref{fact:frustrate-XOR}, we see that the maximum in \eqref{eq:mono-S} is given by
\begin{eqnarray}
\max_{\{P(a,b,c|x,y,z)\} \in \mathcal{S}} \bigg[\omega(G_{AB}) + \omega(G_{AC}) + \omega(G_{BC}) \bigg] = \frac{2}{m^3} \left(|S_{AAA}| + |S_{ACC}|\right) + \frac{3}{m^3} \left(|S_{CCC}| + |S_{AAC}| \right),
\end{eqnarray}
where 
\begin{eqnarray}
S_{AAA} &:=& \left\{ (x,y,z) \in X \times Y \times Z \big| f(x,y) = f(y,z) = f(x,z) = 1 \right\}, \nonumber \\ 
S_{ACC} &:=& \left\{ (x,y,z) \in X \times Y \times Z \big| (f(x,y), f(y,z), f(x,z)) \in \{ (1,0,0), (0,1,0), (0,0,1) \} \right\}, \nonumber \\ 
S_{AAC} &:=& \left\{ (x,y,z) \in X \times Y \times Z \big| (f(x,y), f(y,z), f(x,z)) \in \{ (1,1,0), (1,0,1), (0,1,1) \} \right\}, \nonumber \\ 
S_{CCC} &:=& \left\{ (x,y,z) \in X \times Y \times Z \big| f(x,y) = f(y,z) = f(x,z) = 0 \right\}.
\end{eqnarray}
\end{proof}
For example, for a game such as CHSH where the players receive inputs $x,y \in \{0,1\}$ with $f(0,0) = f(0,1) = f(1,0) = 0$ and $f(1,1) = 1$ we have that $S_{CCC} = \{ (0,0,0), (0,1,0), (1,0,0), (0,0,1) \}$, $S_{AAA} = \{(1,1,1) \}$, $S_{AAC} = \emptyset$ and $S_{ACC} = \{ (1,1,0), (1,0,1), (0,1,1) \}$. We therefore calculate that 
\begin{eqnarray}
\max_{\{P(a,b,c|x,y,z)\} \in \mathcal{S}} \bigg[\omega(CHSH_{AB}) + \omega(CHSH_{AC}) + \omega(CHSH_{BC}) \bigg] = \frac{1}{8}(2 \cdot 4 + 3 \cdot 4) = \frac{5}{2}. 
\end{eqnarray}

Interestingly, monogamy relations hold in signalling theories even for games that admit a classical winning strategy for two players. For example, consider the classically winnable two-player game $G^{cl}$ with inputs $x,y \in \{0,1\}$ and with $f(0,0) = f(0,1) = 0$ and $f(1,0) = f(1,1) = 1$. It is clear that a classical winning strategy exists for this game in the two-player scenario, for instance Alice outputs $a = x$ and Bob outputs $b = 0$ for any input $y$. 
We have that $S_{CCC} = \{(0,0,0), (0,0,1)\}$, $S_{AAA} = \{(1,1,1), (1,1,0)\}$, $S_{ACC} = \{(0,1,0), (0,1,1)\}$ and $S_{AAC} = \{(1,0,0), (1,0,1) \}$. We therefore calculate that 
\begin{eqnarray}
\max_{\{P(a,b,c|x,y,z)\} \in \mathcal{S}} \bigg[\omega(G^{cl}_{AB}) + \omega(G^{cl}_{AC}) + \omega(G^{cl}_{BC}) \bigg] = \frac{1}{8}(2 \cdot 4 + 3 \cdot 4) = \frac{5}{2}. 
\end{eqnarray}

Finally, let us also note the existence of monogamies of two-player games in signalling theories for specific inputs of the third player. For instance, we might consider an expression such as
\begin{equation}
\max_{\{P(a,b,c|x,y,z)\} \in \mathcal{S}} \bigg[\omega(G_{AB})\vert_{z=0} + \omega(G_{AC})\vert_{y=0} + \omega(G_{BC})\vert_{x=0} \bigg],
\end{equation} 
where we consider the value of the two-player game when the third player receives the particular input $0$. In such an optimization, it is possible to violate the monogamy relations derived above. Consider for instance the CHSH game where the players receive inputs $x,y \in \{0,1\}$ with $f(0,0) = f(0,1) = f(1,0) = 0$ and $f(1,1) = 1$. We can calculate by means of an integer linear program that 
\begin{equation}
\max_{\{P(a,b,c|x,y,z)\} \in \mathcal{S}} \bigg[\omega(CHSH_{AB})\vert_{z=0} + \omega(CHSH_{AC})\vert_{y=0} + \omega(CHSH_{BC})\vert_{x=0} \bigg] = 3.
\end{equation}
An optimal strategy achieving this value is the signalling behavior $\{P(a,b,c|x,y,z) \}$ given as $P(0,0,0|0,0,0) = P(0,0,0|0,0,1) = P(0,0,0|0,1,0) = P(0,0,1|0,1,1) = 1$ and $P(1,1,1|1,0,0) = P(1,1,0|1,0,1) = P(1,0,1|1,1,0) = P(1,1,1|1,1,1) = 1$.



\subsection{Jamming of $n$-particle correlations}

Let $\M$ be a spacetime with at least two spatial dimensions. For any $n \geq 2$, it is always possible to arrange a configuration of spacetime points $p, q_1, \ldots, q_n$ in such a way that
\begin{align}
\label{geom0}
p \prec_o \bq \quad \textnormal{but} \quad p \not\prec_o \bq^F \textnormal{ for any } F \subsetneq \{1,\ldots,n\}.
\end{align}
In other words, one can always arrange a scenario in which a single agent may jam $n$-party correlations, but not $k$-party correlations for any $k < n$. Observe that, in fact, it suffices to prove the latter property for $k = n-1$, so from now on we shall assume $F = \{j_0\}^\complement$.

Below we shall construct such a configuration of spacetime points in a $(1+2)$-dimensional Minkowski spacetime $\R^{1,2}$. Since every spacetime is locally diffeomorphic to a Minkowski spacetime, one can then embed this configuration into any given $\M$.

Let $p = (h,0,0)$, where $h \in (0,1)$ is yet unspecified, and define $q_j = (0,\cos(2j\pi/n), \sin(2j\pi/n))$ for $j=1,\ldots,n$. Let $\itSigma_t$ denote the timeslice $\{t\} \times \R^2$. The condition $p \prec_o \bq$ is equivalent to $\cap_j J^+(q_j) \subset R^+(p)$ (cf. \eqref{rel2}), which can in turn be expressed timeslice-wise as
\begin{align*}
\forall t \geq 0 \quad \bigcap\nolimits_j \itSigma_t \cap J^+(q_j) \subset \itSigma_t \cap R^+(p),
\end{align*}
where we have restricted to $t \geq 0$ because otherwise the sets $\itSigma_t \cap J^+(q_j)$ are all empty.

Let $D((x,y),r) \subset \R^2$ denote the open disc of radius $r$ centered at $(x,y)$ and let $\overline{D}((x,y),r)$ denote its closed counterpart. Observe that, for any $t \geq h > 0$,
\begin{align*}
\itSigma_t \cap J^+(q_j) = \{t\} \times \overline{D}\big( (\cos\tfrac{2j\pi}{n}, \sin\tfrac{2j\pi}{n}),t\big) \quad \textnormal{ and } \quad \itSigma_t \cap R^+(p) = \{t\} \times D((0,0),t-h),
\end{align*}
so ensuring that $p \prec_o \bq$ boils down to finding $h$ such that
\begin{align}
\label{geom1}
\forall t \geq 1 \quad \bigcap\nolimits_j \overline{D}\big( (\cos\tfrac{2j\pi}{n}, \sin\tfrac{2j\pi}{n}),t\big) \subset D((0,0),t-h),
\end{align}
where we have further restricted to $t \geq 1$, because for $t < 1$ the intersection on the left-hand side is empty and the inclusion is trivially true.

Similarly, the condition $p \not\prec_o \bq^F$ for any $F = \{j_0\}^\complement$ can be reduced to
\begin{align}
\label{geom2}
\forall j_0 \ \exists t \geq 1 \quad \bigcap\nolimits_{j \neq j_0} \overline{D}\big( (\cos\tfrac{2j\pi}{n}, \sin\tfrac{2j\pi}{n}),t\big) \not\subset D((0,0),t-h).
\end{align}

\begin{figure}[H]
\begin{center}
\resizebox{0.48\textwidth}{!}{\includegraphics[scale=1]{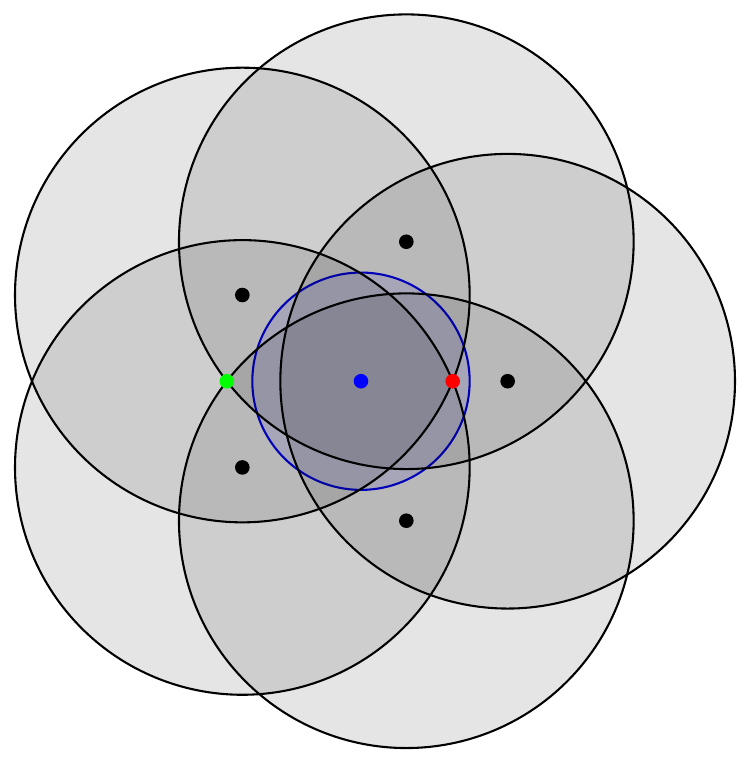}} \qquad
\caption{\label{fig:njam} Illustration of conditions (\ref{geom1},\ref{geom2}) for $n=5$. The black dots lie at $(\cos\tfrac{2j\pi}{n}, \sin\tfrac{2j\pi}{n})$ for $j=1,\ldots,n$, whereas the blue dot lies at $(0,0)$. Condition \eqref{geom1} says that the intersection of all gray discs of radius $t$ (the central darkest region) must lie within the blue disc of radius $t-h$ for all $t\geq 1$. This is equivalent to demanding that the radial coordinate of the red dot, which can be found to be $\sqrt{t^2+\cos\tfrac{2\pi}{n}-2\cos\tfrac{\pi}{n}\sqrt{t^2-\sin^2\tfrac{\pi}{n}}}$ (for $t\geq 1$), is less than $t-h$, leading to \eqref{geom3}. On the other hand, condition \eqref{geom2} says that the intersection of all gray discs but one (the central darkest region with one of the adjacent ``petals'') must stick out of the blue disc for some $t\geq 1$. This is equivalent to demanding that the radial coordinate of the green dot, which can be found to be $\sqrt{t^2-\sin^2\tfrac{2\pi}{n}} - \cos\tfrac{2\pi}{n}$ (for $t\geq 1$), is larger than $t-h$ for some $t$, yielding \eqref{geom4}.}
\end{center}
\end{figure}

After some tedious but straightforward calculations involving finding the coordinates of the intersection points of appropriate circles (see Fig. \ref{fig:njam}), condition \eqref{geom1} can be equivalently expressed as
\begin{align}
\label{geom3}
\forall t \geq 1 \quad h < t - \sqrt{t^2+\cos\tfrac{2\pi}{n}-2\cos\tfrac{\pi}{n}\sqrt{t^2-\sin^2\tfrac{\pi}{n}}} \cv f(t).
\end{align}
But since the function $f(t)$ is decreasing\footnote{One can indeed show that $f'(t) < 0$ for all $t>1$ by direct (albeit lengthy) calculations, which we skip here.} on $[1,\infty)$, this is in turn equivalent to 
\begin{align}
\label{geom4}
h \leq \lim_{t \rightarrow +\infty} f(t) = \cos\tfrac{\pi}{n}.
\end{align}
Similarly, after some analytical geometry, condition \eqref{geom2} can be equivalently written as
\begin{align}
\label{geom5}
\exists t \geq 1 \quad h > t + \cos\tfrac{2\pi}{n} - \sqrt{t^2-\sin^2\tfrac{2\pi}{n}} \cv g(t).
\end{align}
But since the function $g(t)$ is nonincreasing, this is in turn equivalent to
\begin{align}
\label{geom6}
h > \lim_{t \rightarrow +\infty} g(t) = \cos\tfrac{2\pi}{n}.
\end{align}
All in all, we have thus obtained that for any $h \in (\cos\tfrac{2\pi}{n}, \cos\tfrac{\pi}{n}]$ the above configuration of spacetime points indeed satisfies condition \eqref{geom0}.

\hrulefill

\bibliographystyle{naturemag}
\bibliography{ERC}

\hrulefill

\end{document}